\documentclass[11pt,letterpaper]{article}
\usepackage[margin=1in]{geometry}
\usepackage{graphicx} 
\usepackage{authblk}

\usepackage{biblatex} 
\usepackage{geometry}

\usepackage{amsmath}
\usepackage{amssymb}
\usepackage{amsthm}
\usepackage{mathtools}
\usepackage{url}
\usepackage[all]{xy}
\usepackage{latexsym}
\usepackage{xspace}

\makeatletter
\def\slashedarrowfill@#1#2#3#4#5{%
  $\m@th\thickmuskip0mu\medmuskip\thickmuskip\thinmuskip\thickmuskip
  \relax#5#1\mkern-7mu%
  \cleaders\hbox{$#5\mkern-2mu#2\mkern-2mu$}\hfill
  \mathclap{#3}\mathclap{#2}%
  \cleaders\hbox{$#5\mkern-2mu#2\mkern-2mu$}\hfill
  \mkern-7mu#4$%
}
\def\rightslashedarrowfill@{%
  \slashedarrowfill@\relbar\relbar\mapstochar\rightarrow}
\newcommand\xslashedrightarrow[2][]{%
  \ext@arrow 0055{\rightslashedarrowfill@}{#1}{#2}}
\makeatother

\def\slashedrightarrow{\xslashedrightarrow{}}

\newtheorem{theorem}{Theorem}[section]
\newtheorem{lemma}[theorem]{Lemma}
\newtheorem{corollary}[theorem]{Corollary}
\newtheorem{conjecture}{Conjecture}[section]
\newtheorem{problem}{Problem}[section]

\newtheorem{definition}{Definition}[section]
\newtheorem{example}{Example}[section]

\newtheorem{remark}{Remark}[section]

\newcommand{\group}[1]  {
  \mathbb{#1}
}

\newcommand{\catw}[1]  {
  \mathbf{#1}
}
\newcommand{\word}[1]  {
  \mathit{#1}
}
\newcommand{\mor}[3]  {
  #1 \colon #2 \rightarrow #3
}
\newcommand{\tuple}[1]  {
	\langle #1 \rangle
}
\newcommand{\cont}[1]  {
  \catw{Cont}(\group{#1})
}
\newcommand{\classifying}[1]  {
  \catw{ZFA}[#1]
}
\newcommand{\aut}[1]  {
  \mathit{Aut}(#1)
}

\newcommand{\struct}[1]  {
  \mathcal{#1}
}

\newcommand{\boolalg}[1]  {
  \catw{Bool}_{\catw{Set}[1]}
}

\usepackage{tikz}
\usetikzlibrary{automata,positioning}

\usetikzlibrary{calc}
\usetikzlibrary{arrows, decorations.markings}
\usetikzlibrary{patterns}
\tikzstyle{vecArrow} = [thick, decoration={markings,mark=at position
   1 with {\arrow[semithick]{open triangle 60}}},
   double distance=1.4pt, shorten >= 5.5pt,
   preaction = {decorate},
   postaction = {draw,line width=1.4pt, white,shorten >= 4.5pt}]
   
\tikzstyle{vecArrow2} = [thick, black, double distance=1.4pt, shorten >= 1.5pt,
preaction = {decorate},
   postaction = {draw,line width=1.4pt, white,shorten >= 1.5pt}]
   
\tikzstyle{innerWhite} = [semithick, white,line width=1.4pt, shorten >= 4.5pt]

\title{A note on Stone-Čech compactification in ZFA}
\author{Michał R.~Przybyłek}
\affil{School of Informatics, The University of Edinburgh\\Polish-Japanese Academy of Information Technology}
\date{}

\begin{document}
\maketitle

\begin{abstract}
Working inside Zermelo-Fraenkel Set Theory with Atoms over an $\omega$-categorical $\omega$-stable structure we provide a structure theorem for Stone-Čech compactification of definable sets. In particular, we prove that the Stone-Čech compactification of a definable set is definable, which allows us to encode some \emph{infinitary} constructions over definable sets as \emph{finitary} ones -- we show that for a definable set $X$ with its Stone-Čech compactification $\overline{X}$ the following holds: a) the powerset $\mathcal{P}(X)$ of $X$ is isomorphic to the finite-powerset $\mathcal{P}_{\textit{fin}}(\overline{X})$ of $\overline{X}$, b) the vector space $\mathcal{K}^X$ over a field $\mathcal{K}$ is the free vector space $F_{\mathcal{K}}(\overline{X})$ on $\overline{X}$ over $\mathcal{K}$, c) every probability measure on $X$ is tantamount to a \emph{discrete} measure on $\overline{X}$. This leads to some new results about equivalence of certain computational problems.

\end{abstract}

\section{Introduction}
\label{sec:introduction}

It is an old observation that goes back to Stanisław Ulam that one can separate ``small sets'' from ``large sets'' and ``large sets'' from ``very large sets'' by the existence of certain ultrafilters on the sets. For example, let us work in classical mathematics ZFC. Then a set is finite if and only if it is in a bijective correspondence with the set of ultrafilters on it, in which case, every ultrafilter is principal. Therefore, we may say that a set is \emph{infinite} if there is a non-principal ultrafilter on it\footnote{Of course, we do not need the full power of the Axiom of Choice, Boolean Prime ideal Theorem is sufficient. Note, however, that it is consistent with ZF \cite{blass1997model} and even with ZF+DC+Hahn-Banach Theorem \cite{pincus1977definability} that \emph{all} ultrafilters are principal.}. One may also ask about the existence of non-principal countably-additive ultrafilters on a set and it is well-known that the smallest set having such an ultrafilter\footnote{If it exists, because its existence is not provable from ZFC alone.} must be strongly inaccessible (therefore, it must be ``very large'', as the sets below it form an inner model of ZFC).

A main theme of this paper is the structure of ultrafilters on definable sets in Fraenkel-Mostowski permutational models of Set Theory with Atoms (ZFA). In this setting the Axiom of Choice fails (unless the permutational model is trivial), and the Boolean Prime Ideal Theorem (BPIT) may hold or fail, but, counter-intuitively, it is mostly irrelevant for our results. In fact, our main results concern permutational models over $\omega$-categorical $\omega$-stable structures (although we will discuss other structures in the paper), in which case BPIT fails for general Boolean algebras, but holds for power-set algebras (see Theorem~\ref{t:ultrafilters:zfa} from Appendix~\ref{sec:app:intro:theorem}). Examples of such structures include Example~\ref{e:pure:sets} and Example~\ref{e:vector:space}, but not Example~\ref{e:pure:sets:constants}, Example~\ref{e:rationals} nor Example~\ref{e:graph}.

\begin{example}[Pure sets]\label{e:pure:sets}
Let $\struct{N} = \{0,1,2,\dotsc\}$ be a countably infinite set over empty signature $\Xi$. Then the first order theory of $\struct{N}$ is $\omega$-categorical and $\omega$-stable, i.e.~there is exactly one model of the theory up to an isomorphism for every infinite cardinal number. This theory is called the theory of ``pure sets''.
\end{example}

\begin{example}[Pure sets with constants]\label{e:pure:sets:constants}
Let $\struct{N} \sqcup N$ be the structure from Example~\ref{e:pure:sets} over an extended signature consisting of all constants $n \in N$. Then the first order theory of $\struct{N} \sqcup N$  has countably many non-isomorphic countable models, therefore is not $\omega$-categorical. It is, however, $\omega$-stable, because adding countably many constants cannot change the stability of a structure.
\end{example}

\begin{example}[Vector space over a finite field]\label{e:vector:space}
Let $V_{\mathcal{F}}$ be the free $\aleph_0$-dimensional vector space over a finite field $\mathcal{F}$. We shall consider $V_{\mathcal{F}}$ with its natural vector-space structure, i.e.~$\struct{V}_{\mathcal{F}} = \tuple{V_{\mathcal{F}}, {+}, (-)r}$ for every $r \in \mathcal{F}$. This theory is both $\omega$-categorical and $\omega$-stable, because for every infinite cardinal $\kappa$ it has exactly one model (up to isomorphism) of cardinality $\kappa$ --- the free vector space on $\kappa$ base vectors.  
\end{example}

\begin{example}[Rational numbers with ordering]\label{e:rationals}
Let $\struct{Q} = \tuple{Q, {\leq}}$ be the structure whose universe is interpreted as the set of rational numbers $Q$ with a single binary relation ${\leq} \subseteq Q \times Q$ interpreted as the natural ordering of rational numbers. Then the first order theory of $\struct{Q}$ is $\omega$-categorical but not $\omega$-stable.
\end{example}

\begin{example}[Random graph]\label{e:graph}
Let $\struct{R}$ be a countable graph over signature consisting of a single binary relation $E$ and satisfying the following two axioms: (Simplicity Axiom) $R$ is symmetric and irreflexive; (Extension Axiom) if $V_0, V_1 \subset R$ are finite disjoint subsets, then there is $v \in R$ such that for every $v_0 \in V_0$ the relation $R(v, v_0)$ holds and for every $v_0 \in V_0$ the relation $R(v, v_1)$ does not hold. Structure $\struct{R}$ is $\omega$-categorical, but not $\omega$-stable.
\end{example}

Interestingly, definable sets in $\omega$-categorical $\omega$-stable structures behave like something intermediate between ``small sets'' and ``large sets'' --- they enjoy many closure properties of finite sets, but the closure operators deviate significantly from the identity.

First of all, in classical ZFC, the distinction between ``small sets'' and ``large sets'' is not only a matter of a mere existence of non-principal ultrafilters, i.e.~``large sets'' have an enormous number of non-principal ultrafilters, whereas small sets have none. That is, for a set $X$ the number of non-principal ultrafilters is either $0$ (in case $X$ is finite) or doubly-exponential: $2^{2^X}$ (in case $X$ is infinite). In contrast (see Theorem~\ref{t:stone:cech:compactification}), the number of non-principal ultrafilters on definable sets in our permutational models may be bounded by a polynomial. In fact, the set of ultrafilters on a definable set is always definable. For example, in the basic Fraenkel-Mostowski model, the set of atoms $N$ has only one non-principal ultrafilter (consisting of all cofinite subsets of $N$), and for the set of distinct pairs of atoms $N^{[2]}$, we have exactly $2N + 1$ non-principal ultrafilters.

Secondly, in classical ZFC, a vector space is isomorphic to its dual if and only if it is finite dimensional. Let us assume for simplicity that our base field is $2$. Then if $V$ is an infinite-dimensional vector space with a base $X$, then the dimension of its dual space grows exponentially in $X$: i.e.~the dimension of $2^X$ is exactly $2^{X}$. Therefore, the base of $2^X$ is isomorphic to the set of ultrafilters on $X$ if and only if $X$ is finite-dimensional. In contrast, for every definable set $X$ in our permutational models, the set of ultrafilters on $X$ is isomorphic to the base of $2^X$, which proves that dual spaces have basis and gives an explicit construction of the basis (see Theorem~\ref{t:free:space}) Moreover, since the space $2^X$ is just the power set $\mathcal{P}(X)$ of $X$ and the free vector space on a set is just the set $\mathcal{P}_{\word{fin}}(X)$ of finite subsets of $X$, Theorem~\ref{t:free:space} implies that for every definable $X$ we have that $\mathcal{P}(X) \approx \mathcal{P}_{\word{fin}}(Y)$ for some definable $Y$, i.e. $Y$ can be taken to be the set of ultrafilters on $X$. This means, that we can effectively, transfer theorems about \emph{finite} subsets of definable sets to \emph{all} subsets of definable sets. For some of the applications, see Subsection~\ref{sub:sec:machines} below.

Finally, in classical ZFC, a set $X$ is finite if and only if every measure $\mu$ on the full algebra of all subsets of $X$ is a \emph{finite} combination of mass-measures, i.e.~$\mu = \sum_{i=1}^n r_i x_i$, where $\sum_{i=1}^n r_i = 1$, each $r_i$ is positive, and $x_i$ is concentrated on a singleton. Of course, a mass measure on a set is just a principal ultrafilter on the set. Moreover, every countably-additive ultrafilter is tantamount to a measure taking values in $\{0, 1\}$. But for definable sets $X$ in our permutational models, being countably-additive is a vacuous condition, because every countable collection of subsets of $X$ must be essentially finite. 
Therefore, every ultrafilter on a definable set is tantamount to a $\{0, 1\}$-measure. As it turns out, every measure on a definable set is a finite combination of ultrafilters on the set (see Theorem~\ref{t:definable:measure}).

\subsection{Preliminaries}
In this section we fix our terminology and notation. We assume that the reader is familiar with basic concepts from category theory \cite{johnstone2003sketches} \cite{kelly1982basic}, model theory\cite{chang1990model} \cite{hodges1993model} and set theory \cite{halbeisen2017combinatorial}, \cite{jech1977axiom} \cite{jech2008axiom}. Sets will be usually denoted by capital Roman letters $A, B, X, Y$ etc. Infinite ordinals will be denoted by lower case Greek letters $\alpha, \beta, \lambda, \omega$, etc. Finite ordinal numbers will be denoted by lower case Roman letters $m, n, k$, etc. By convention we shall identify subsets $A_0 \subseteq A$ with their characteristic functions $\mor{A_0}{A}{2}$, so $x \in A_0$ is the same as $A_0(x) = 1$

Throughout the paper we will consider models of a complete countable single-sorted first-order theory with no finite models. Formulas will be denoted by lower case Greek letters $\psi, \phi, \theta, \dots$. We shall write $\phi(x_1, x_2, \dotsc, x_n)$ to indicate that the free variables in $\phi$ are in $x_1, x_2, \dotsc, x_n$. We will also write $\overline{x}$ for the sequence $x_1, x_2, \dotsc, x_n$ and then $|\overline{x}|$ for $n$ -- the length of the sequence. A sentence is a formula without free variables. Structures will be denoted by stylised capital Roman letters $\struct{A}$, $\struct{B}$, $\struct{N}$, etc. If $\struct{A}$ is a structure then its universum will be denoted by $A$. The elements of $A$ should be thought of as the ``atoms''. If $\phi(\overline{x}, \overline{y})$ is a formula and $\overline{a} \in A_0^{|\overline{y}|}$ is a sequence of elements in $A_0$ for some $A_0 \subseteq A$, then we call  $\phi(\overline{X}, \overline{a})$ a formula with parameters in $A_0$ or just formula with parameters in case $A_0 = A$. A (complete) $n$-type over $A_0 \subseteq A$ is just the maximal consistent set of formulas $\phi(\overline{x}, \overline{a})$ with $n$ free variables, i.e.~$|\overline{x}| = n$, and parameters $\overline{a}$ from $A_0$. Types will be usually denoted by lower case Roman letters $p, q, r, \dotsc$. The set of all $n$-types over $A_0$ will be denoted by $S_n(A_0)$. A type $p \in S_n(A_0)$ is \emph{definable} over $B_0 \subseteq A$ if for every $\phi(\overline{x}, \overline{y})$ there exists a formula $\psi(\overline{y}, \overline{b})$ with parameters in $B_0$ such that $\phi(\overline{x}, \overline{a}) \in p \Leftrightarrow \psi(\overline{a}, \overline{b})$. A type is definable if it is definable for some $B_0 \subseteq A$ and it is finitely definable if $B_0$ is finite.
We say that a set $D \subseteq A^n$ is \emph{definable} with parameters $B_0 \in A$ if there exists a formula $\phi(\overline{x}, \overline{b})$ such that $D = \{\overline{a} \in A^n \colon \phi(\overline{a}, \overline{b})\}$. A set defined by a formula $\phi(\overline{x}, \overline{b})$ will be denoted by $\phi(A, \overline{b})$. A $D \subseteq A^n$ is \emph{definable} if it is definable for some $B_0 \subseteq A$. Notice that for a complete theory two formulas are equivalent if and only if they define the same set (for any model of the theory). A theory is said to be $\omega$-categorical if for every natural $n$ the set $S_n(\emptyset)$ of $n$-types without parameters is finite. Equivalently, if for every natural $n$ there are only finitely many formulas $\phi(x_1, x_2, \dotsc, x_n)$ modulo the theory. A theory is said to be $\omega$-stable if for every natural $n$ the set $S_n(A)$ of $n$-types over universum of the model $A$ is countable. 

We shall speak about Morley rank and Morley degree of a formula in a few contexts. Morley rank together with Morley degree associate with every formula $\phi(\overline{x}, \overline{a})$ with parameters an invariant playing the role of a generalised dimension. Morley rank of a formula consists of a generalised ordinal number $\alpha$, which can be either $-1$, an ordinal number or $\infty$ symbol (i.e.~unbounded dimension). The below definition is inductive and starts by providing an upper bound on Morley rank $\mathit{MR}(\phi)$ of $\phi$. For every $\phi(A, \overline{a}) \neq \emptyset$ we have that $\mathit{MR}(\phi) \geq 0$ and if $\phi(A, \overline{a}) = \emptyset$ then we set $\mathit{MR}(\phi) = 0$. If $\alpha$ is a limit ordinal, then $\mathit{MR}(\phi) \geq 0$ if an only if $\mathit{MR}(\phi) \geq \beta$ for all $\beta < \alpha$. For any ordinal $\alpha$ we have that $\mathit{MR}(\phi) \geq \alpha + 1$ if and only if there is an \emph{infinite} sequence of pairwise disjoint formulas $\psi_1(\overline{x}, \overline{a_1}), \psi_2(\overline{x}, \overline{a_2}), \dotsc$ with parameters such that for every $i$ we have that $\mathit{MR}(\psi_i) \geq \alpha$ and $\phi(A, \overline{a}) = \bigsqcup_{i} \psi_i(A, \overline{a_i})$. Then the Morley rank of formula $\phi$ is defined as the biggest $\alpha$ such that $\mathit{MR}(\phi) \geq \alpha$ or $\infty$ is such an $\alpha$ does not exists. If the Morley Rank of $\phi(\overline{x}, \overline{a})$ is an \emph{ordinal number} $\alpha$ then we define the Morley degree $\mathit{MR}(\phi)$ of $\phi$ to be the greatest natural number $k$ such that there are $k$ pairwise disjoint formulas $\psi_i(\overline{x}, \overline{a_i})$ with Morley rank $\alpha$ such that $\phi(A, \overline{a}) = \bigsqcup_{i} \psi_i(A, \overline{a_i})$. It is a standard result of Model Theory that in a $\omega$-categorical $\omega$-stable theory, every formula has an integer Morley rank.

Let $\struct{A}$ be an algebraic structure (both operations and relations are allowed) with universum $A$. We shall think of elements of $\struct{A}$ as ``atoms''. A von Neumann-like hierarchy $V_\alpha(\struct{A})$ of sets with atoms $\struct{A}$ can be defined by transfinite induction \cite{mostowski1939unabhangigkeit}, \cite{halbeisen2017combinatorial}:
\begin{itemize}
    \item $V_0(\struct{A}) = A$
    \item $V_{\alpha + 1}(\struct{A}) = \mathcal{P}(V_{\alpha}(\struct{A})) \cup V_{\alpha}(\struct{A})$
    \item $V_{\lambda}(\struct{A}) = \bigcup_{\alpha < \lambda} V_{\alpha}(\struct{A})$ if $\lambda$ is a limit ordinal
\end{itemize}
Then the cumulative hierarchy of sets with atoms $\struct{A}$ is just $V(\struct{A}) = \bigcup_{\alpha \colon \mathit{Ord}} V_{\alpha}(\struct{A})$. Observe, that the universe $V(\struct{A})$ carries a natural action $\mor{(\bullet)}{\aut{\struct{A}} \times V(\struct{A})}{V(\struct{A})}$ of the automorphism group $\aut{\struct{A}}$ of structure $\struct{A}$ --- it is just applied pointwise to the atoms of a set. If $X \in V(\struct{A})$ is a set with atoms then by its set-wise stabiliser we shall mean the set: $\aut{\struct{A}}_X = \{\pi \in \aut{\struct{A}} \colon \pi \bullet X = X \}$; and by its point-wise stabiliser the set: $\aut{\struct{A}}_{(X)} = \{\pi \in \aut{\struct{A}} \colon \forall_{x \in X} \pi \bullet x = x \}$. Moreover, for every $X$, these sets inherit a group structure from $\aut{\struct{A}}$.

There is an important sub-hierarchy of the cumulative hierarchy of sets with atoms $\struct{A}$, which consists of ``symmetric sets'' only. To define this hierarchy, we have to equip $\aut{\struct{A}}$ with the structure of a topological group. A set $X \in V(\struct{A})$ is \emph{symmetric} if the set-wise stabilisers of all of its descendants $Y$ is an open set (an open subgroup of $\aut{\struct{A}}$), i.e.~for every $Y \in^* X$ we have that: $\aut{\struct{A}}_{Y}$ is open in $\aut{\struct{A}}$, where ${\in^*}$ is the reflexive-transitive closure of the membership relation ${\in}$. A function between symmetric sets is called symmetric if its graph is a symmetric set.
Of a special interest is the topology on $\aut{\struct{A}}$ inherited from the product topology on $\prod_A A = A^A$ (i.e.~the Tychonoff topology). We shall call this topology the canonical topology on $\aut{\struct{A}}$. In this topology, a subgroup $\group{H}$ of $\aut{\struct{A}}$ is open if there is a finite $A_0 \subseteq A$ such that: $\aut{\struct{A}}_{(A_0)} \subseteq \group{H}$, i.e.: group $\group{H}$ contains a pointwise stabiliser of some finite set of atoms. The sub-hierarchy of $V(\struct{A})$ that consists of symmetric sets according to the canonical topology on $\aut{\struct{A}}$ will be denoted by $\catw{ZFA}(\struct{A})$ (it is a model of Zermelo-Fraenkel set theory with atoms).

\begin{remark}
The above definition of hierarchy of symmetric sets is equivalent to another one used in model theory. 
By a normal filter of subgroups of a group $\group{G}$ we shall understand a filter $\mathcal{F}$ on the poset of subgroups of $\group{G}$ closet under conjugation, i.e.~if $g \in \group{G}$ and $\group{H} \in \mathcal{F}$ then $g \group{H} g^{-1} = \{g \bullet h \bullet g^{-1} \colon h \in \group{H}\} \in \mathcal{F}$. Let $\mathcal{F}$ be a normal filter of subgroups of $\aut{\struct{A}}$. We say that a set $X \in V(\struct{A})$ is $\mathcal{F}$-symmetric if the set-wise stabilisers of all of its descendants $Y$ belong to $\mathcal{F}$ --- i.e. $Y \in^* \mathcal{F}$. To see that the definitions of symmetric sets and $\mathcal{F}$-symmetric sets are equivalent, observe first that if $\group{G}$ is a topological group, then the set $\mathcal{F}$ of all open subgroups of $\group{G}$ is a normal filter of subgroups. In the other direction, if $\mathcal{F}$ is a normal filter of subgroups of a group $\group{G}$, then we may define a topology on $\group{G}$ by declaring sets $U \subseteq \group{G}$ to be open if they satisfy the following property: for every $g \in U$ there exists $\group{H} \in \mathcal{F}$ such that $g \group{H} \subseteq U$. According to this topology a group $\group{U}$ is open iff  $\group{U} \in \mathcal{F}$ --- just observe that for every group $\group{U}$ and for every $g \in \group{U}$ we have that $g\group{U} = \group{U}$; and if $\group{H} \in \mathcal{F}$ such that $\group{H} = 1\group{H} \subseteq \group{U}$ then by the property of the filter, $\group{U} \in \mathcal{F}$.    
\end{remark}

\begin{example}[The basic Fraenkel-Mostowski model]\label{e:first:zfa}
Let $\struct{N}$ be the structure from Example~\ref{e:pure:sets}. We call $\catw{ZFA}(\struct{N})$ the basic Fraenkel-Mostowski model of set theory with atoms. Observe that $\aut{\struct{N}}$ is the group of all bijections (permutations) on $N$. 
The following are examples of sets in $\catw{ZFA}(\struct{N})$:
\begin{itemize}
\item all sets without atoms, e.g.~$\emptyset, \{\emptyset\}, \{\emptyset, \{\emptyset\}, \dotsc\}, \dotsc$
\item all finite subsets of $N$, e.g.~$\{0\}, \{0,1,2,3\}, \dotsc$
\item all cofinite subsets of $N$, e.g.~$\{1, 2, 3, \dotsc\}, \{4, 5, 6, \dotsc\}, \dotsc$
\item $N\times N$
\item $\{\tuple{a,b} \in N^2 \colon a \neq b\}$
\item $N^* = \bigcup_{k\in N} N^k$ 
\item $\mathcal{P}_{\word{fin}}(N) = \{N_0 \colon N_0 \subseteq N, \textit{$N_0$ is finite}\}$
\item $\mathcal{P}(N) = \{N_0 \colon N_0 \subseteq N, \textit{$N_0$ is symmetric}\}$
\end{itemize}
\end{example}

\begin{example}[The ordered Fraenkel-Mostowski model]\label{e:ordered:zfa}
Let $\struct{Q}$ be the structure from Example~\ref{e:rationals}. We call $\catw{ZFA}(\struct{Q})$ the ordered Fraenkel-Mostowski model of set theory with atoms. Observe that $\aut{\struct{Q}}$ is the group of all order-preserving bijections on $Q$. 
All symmetric sets from Example~\ref{e:first:zfa} are symmetric sets in $\catw{ZFA}(\struct{Q})$ when $N$ is replaced by $Q$. Here are some further symmetric sets:
\begin{itemize}
\item $Q^{<2} = \{\tuple{p, q} \in Q^2 \colon p \leq q\}$
\item $Q^{<2} \cap [0, 1]^2 = \{\tuple{p, q} \in Q^2 \colon 0 \leq p \leq q \leq 1 \}$
\end{itemize}
\end{example}

Observe that the group $\aut{\struct{A}}_{(A_0)}$ is actually the group of automorphism of structure $\struct{A}$ extended with constants $A_0$, i.e.: $\aut{\struct{A}}_{(A_0)} = \aut{\struct{A} \sqcup A_0}$. Then a set $X \in V(\struct{A})$ is symmetric if and only if there is a finite $A_0 \in A$ such that $\aut{\struct{A} \sqcup A_0} \subseteq \aut{\struct{A}}_X$ and the canonical action of topological group $\aut{\struct{A} \sqcup A_0}$ on discrete set $X$ is continuous. A symmetric set is called $A_0$-equivariant (or equivariant in case $A_0 = \emptyset$) if $\aut{\struct{A} \sqcup A_0} \subseteq \aut{\struct{A}}_X$. Therefore, the (non-full) subcategory of $\catw{ZFA}(\struct{A})$ on $A_0$-equivariant sets and $A_0$-equivariant functions (i.e.~functions whose graphs are $A_0$-equivariant) is equivalent to the category $\cont{\aut{\struct{A} \sqcup \word{A_0}}} \subseteq \catw{Set}^{\aut{\struct{A} \sqcup \word{A_0}}}$ of continuous actions of the topological group $\aut{\struct{A} \sqcup A_0}$ on discrete sets. We will heavily use the transfer principle developed in \cite{licsMRP}, which is based on the observation that adding finitely many constants to an $\omega$-categorical and $\omega$-stable structure and closing it under elimination of imaginaries, produces structure, which is $\omega$-categorical and $\omega$-stable.

\begin{definition}[Definable set in ZFA]
We shall say that an $A_0$-equivariant set $X \in \catw{ZFA}(\struct{A})$ is definable if its canonical action has only finitely many orbits, i.e.~if the relation $x \equiv y \Leftrightarrow \exists_{\pi \in \aut{\struct{A} \sqcup \word{A_0}}} \; x = \pi \bullet y$ has finitely many equivalence classes.
\end{definition}

For an open subgroup $\group{H}$ of $\aut{\struct{A}}$ let us denote by $\aut{\struct{A}}/\group{H}$ the quotient set $\{\pi \group{H} \colon \pi \in \aut{\struct{A}} \}$. This set carries a natural continuous action of $\aut{\struct{A}}$, i.e.~for $\sigma, \pi \in \aut{\struct{A}}$, we have $\sigma \bullet \pi \group{H} = (\sigma \circ \pi) \group{H}$. All transitive (i.e.~single orbit) actions of $\aut{\struct{A}}$ on discrete sets are essentialy of this form (see for example Chapter~III, Section~9 of \cite{maclane2012sheaves}). Therefore, equivariant definable sets are essentially finite unions of sets of the form $\aut{\struct{A}}/\group{H}$. Moreover, if structure $\struct{A}$ is $\omega$-categorical (Example~\ref{e:pure:sets}, Example~\ref{e:vector:space}, Example~\ref{e:rationals}, Example~\ref{e:graph}), then equivariant definable sets are the same as sets definable in the first order theory of $\struct{A}$ extended with elimination of imaginaries \cite{licsMRP}. Therefore, we can just speak of definable sets.

\subsection{Some applications of Stone-Čech compactification to register machines}
An important type of automata has been defined by Kaminski and Francez \cite{kaminski1994finite}. The authors called these type of automata ``finite memory machines'', or ``register machines''. A finite memory machine is a finite automaton augmented with a finite number of registers $R_i$ that can store natural numbers. The movement of the machine can depend on the control state, on the letter and on the content of the registers. The dependency on the content of the registers is, however, limited --- the machine can only test for equality (no formulas involving successor, addition, multiplication, etc.~are allowed). Here is a suitable generalisation of this definition to a general structure $\struct{A}$.

A finite memory automata (over structure $\struct{A}$) with $k$ registers over alphabet $\Sigma$ is a quadruple $\tuple{S, \delta, I, F}$ such that:
\begin{itemize}
\item $S$ is a finite set of states
\item $I \subseteq S$ is a set of initial states, and $\phi_I \subseteq A^k$ is a set of possible initial configurations of registers
\item $F \subseteq S$ is a set of final states, and $\phi_F \subseteq A^k$ is a set of possible final configurations of registers
\item $\delta \subseteq (\Sigma \times S \times A^k) \times (S \times A^k)$ is a transition relation such that for every $s, s' \in S$ the relation $\delta(s, s') \subseteq (\Sigma \times A^k) \times A^k$ is $\struct{A}$-definable.
\end{itemize}
A finite memory automata is called deterministic if $I$ is the singleton and the transition relation $\sigma$ is functional.

It is well-known that finite memory automata in the above sense are equivalent to definable automata in ZFA, i.e.~set $S\times A^k$ can be identified with a definable set, and then the transition relation becomes a definable relation between definable sets. Therefore, a definable deterministic automata is just a definable function $\mor{\sigma}{\Sigma \times S}{S}$ between definable sets together with an initial state $s_0 \in S$ and a set of final states $F \subseteq S$. To define the language $L(A)$ recognised by such an automaton, we have to observe that functions $\mor{\sigma}{\Sigma \times S}{S}$ are tantamount to functions $\mor{\sigma^\dag}{\Sigma}{S^S}$ and $S^S$ carries a structure of a monoid under composition of functions $S \rightarrow S$, and so, one may extend $\sigma^\dag$ to the unique homomorphism $\mor{h}{\Sigma^*}{S^S}$ from the free monoid on $\Sigma$ generators. The language of $A$ is just the set $L(A) =\{w \in \Sigma^* \colon h(w)(s_0) \in F \}$. Similarly, the crucial observation needed to define the language of a non-deterministic automaton is that the transition relation $\mor{\sigma}{\Sigma \times S}{\mathcal{P}(S)}$ is tantamount to $\sigma^\dag \colon \Sigma \rightarrow \mathcal{P}(S)^S \approx \mathcal{P}(S\times S)$ and $\mathcal{P}(S\times S)$ carries a monoidal structure induced by the composition of relations $S \slashedrightarrow S$. One may wonder, if we can substitute the powerset operator with other operators on $S$. The answer is yes, provided that the operator is a \emph{strong} monad (this is a sufficient, but not necessary condition) on the category $\classifying{A}$, i.e.~if $\mor{T}{\classifying{A}}{\classifying{A}}$ is a strong monad, then $T(S)^S$ is naturally a monoid under Kleisli composition of functions $S \rightarrow T(S)$.

\subsubsection{On a machine that can erase information from its registers}\label{sub:sec:machines}
Intuitively, erasing information from a register, should make all of the values in the register ``equally likely'' and each individual value ``completely unlikely''. If $R$ can hold a value from $\struct{N}$, then we can model this by assigning values to $R$ in such a way that the probability for $R$ to get values from any finite subset of $\struct{N}$ is zero. This corresponds to the assignment of a value to $R$ at ``random'' according to the only non-principal ultrafilter on $N$, i.e.~the ultrafilter consisting of all cofinite subsets of $N$. This, in turn, suggests that we should model the operation of erasing information from registers via ultrafilter automata: that is, automata for the ultrafilter monad $\mor{\overline{(-)}}{\classifying{A}}{\classifying{A}}$. Notice that in the classical setting of finite automata, we do not speak about ``finite ultrafilter automata'', because every ultrafilter on a finite set is principal. Here is the formal definition.
\begin{definition}[Ultra-automaton]\label{d:ultra:machine}
    A deterministic ultra-automata (or erasing information automata) over a definable alphabet $\Sigma$ consists of a definable set $S$, definable transition relation $\mor{\sigma}{\Sigma \times S}{\overline{S}}$ an initial state $s_0 \in S$ and a set of final states $F$.
\end{definition}
By Theorem~\ref{t:ultrafilter:strong}, the set of ultrafilters carries a strong monad structure, therefore we can define the language of such an automaton in a natural way. 

\begin{definition}[Language of an ultra-automaton]
    The language $L$ of an automaton $\mor{\sigma}{\Sigma \times S}{\overline{S}}$ with initial state $s_0$ and final states $F$ is defined as $L = \{w \in \Sigma^* \colon h(w)(s_0) \in F \}$, where $\mor{h}{\Sigma^*}{\overline{S}^S}$ is the unique homomorphism of monoids extending function $\mor{\sigma^\dag}{\Sigma}{\overline{S}^S}$
\end{definition}
    One may extend the above definition to non-deterministic ultra-automaton by observing that the ultrafilter monad can be extended to internal relations. This is however unnecessary due to the next theorem and its proof. 

\begin{theorem}[On the expressive power of ultra-automata]\label{t:expressive:ultra:automata}
    Let $\struct{A}$ be an $\omega$-categorical and $\omega$-stable structure. The languages in $\classifying{A}$ recognised by definable ultra-automata are exactly the same as the languages recognised by deterministic automata.
\end{theorem}
\begin{proof}
    According to Corollary~\ref{c:definable:monad}, the ultrafilter monad restricts to the monad on definable sets. Thus, $\overline{S}$ is definable. Moreover, because the monad is strong and the structure of the monad is equivariant, every definable function $\mor{\sigma}{\Sigma \times S}{\overline{S}}$ extends to a definable function $\mor{\sigma}{\Sigma \times \overline{S}}{\overline{S}}$.
    Observe also that $\overline{S}^S$ is a submonoid of $\overline{S}^{\overline{S}}$ (actually, the full submonoid on continuous functions), therefore the languages recognised by  $\mor{\sigma}{\Sigma \times S}{\overline{S}}$ and 
    $\mor{\overline{\sigma}}{\Sigma \times \overline{S}}{\overline{S}}$ are the same.
\end{proof}
The above theorem effectively says that we can include the ``erase information'' operation to register machines without changing they properties.

While for general $\omega$-categorical structures the ultrafilter monad do not restrict to definable sets (see Example~\ref{e:types:in:graphs}), we conjecture that Theorem~\ref{t:expressive:ultra:automata} holds for every $\omega$-categorical structure.

\begin{conjecture}
    Let $\struct{A}$ be an $\omega$-categorical structure. The languages in $\classifying{A}$ recognised by definable ultra-automata are exactly the same as the languages recognised by deterministic automata.
\end{conjecture}

\subsubsection{Weighted register machines}
In \cite{BKM21} M.~Bojanczyk, B.~Klin and M.~Moerman introduced and studied weighted definable automata in $\classifying{A}$ for an $\omega$-categorical structure $\struct{A}$. Here is their definition. 

\begin{definition}[Weighted automaton]\label{d:weighted:automata}
    Let us fix a field $\mathcal{K}$. A weighted definable automaton consists of definable sets $S$ and $\Sigma$, called
the states and the alphabet, and symmetric functions:
\begin{itemize}
    \item $\mor{I}{S}{\mathcal{K}}$ for initial states
    \item $\mor{F}{S}{\mathcal{K}}$ for final states
    \item $\mor{\sigma}{\Sigma \times S \times S}{\mathcal{K}}$
\end{itemize}
subject to the following requirement: there are finitely many states with nonzero initial weight, and also for every state $s \in S$ and input letter $a \in \Sigma$, there are finitely many states $s' \in S$ such that the transition $(a, s, s')$ has nonzero weight.
\end{definition}

Obviously, a function $\mor{\sigma}{\Sigma \times S \times S}{\mathcal{K}}$ should be rewritten as $\mor{\sigma}{\Sigma \times S}{\mathcal{K}^S}$ to exhibit more similarities with other types of automata. The problem with this definition is that $\mathcal{K}^S$ is not the free vector space $F(S)$ on $S$, and if try to define Kleisli composition, or equivalently, extend freely  $\sigma$ to the liner map in the second variable, we will obtain: $\mor{F(\sigma)}{\Sigma \times F(S)}{\mathcal{K}^S}$. Generally, it seems that there is no natural way to induce the monoid structure, because functions $F(S) \rightarrow \mathcal{K}^S$ do not compose. One way of dealing with this obstacle is to impose an extra condition on $\sigma$ as in Definition~\ref{d:weighted:automata}. This condition is a convoluted way of saying that function $\mor{\sigma}{\Sigma \times S}{\mathcal{K}^S}$ factors as $\Sigma \times S \rightarrow F(S) \subseteq \mathcal{K}^S$. Because, $F(S)^{S} \approx \mathit{Lin}(F(S), F(S))$ has a natural monoidal structure in the category of vector spaces, such $\sigma$ extends to the linear homomorphism $\mor{\overline{\sigma}}{F(\Sigma^*)}{F(S)^{S}}$, where $F(\Sigma^*) = \bigoplus_{k = 0} F(\Sigma)^k$.
Moreover, the condition on initial states $I$ is a sophisticated way of saying that $I$ is tantamount to a vector $s_0 \in F(S)$ and one may define the language of the automaton to be the restriction of $\mor{F \circ \overline{\sigma}(-, s_0)}{F(\Sigma^*)}{\mathcal{K}}$ to the basis $\Sigma^*$.

However, when structure $\struct{A}$ is $\omega$-stable, \emph{there is a way} to define such a composition. By Theorem~\ref{t:free:space}, vector space $\mathcal{K}^S$ is isomorphic to the free vector space $F(\overline{S})$ on the Stone–Čech compactification $\overline{S}$ of $S$. Therefore, the transition relation $\mor{\sigma}{\Sigma \times S}{\mathcal{K}^S}$ can be rewritten as $\mor{\sigma}{\Sigma \times S}{F(\overline{S})}$ and extended to $\Sigma \times F(\overline{S}) \rightarrow F(\overline{S})$.

Moreover, in Section~\ref{sec:definable:spaces} we prove that for definable set $S$, the space $F(S)^S$ has a definable basis $M \subseteq \overline{S \times S}$, that is: $F(S)^S \approx F(M)$. Therefore, the linear monoid $F(S)^S$ has a definable basis. The concept of a language recognized by a linear monoid is defined in the usual way.
\begin{definition}[Language recognied by a definable linear monoid]\label{d:monoid:language}
    Let $\mathcal{M} = \tuple{M, \bullet, \epsilon}$ be a finitely supported linear monoid with a definable basis. We say that $\mathcal{M}$ recognizes language $\mor{L}{\Sigma^*}{\mathcal{K}}$ if there exists a linear functional $\mor{f}{M}{\mathcal{K}}$ and a homomorphism of monoids $\mor{h}{F(\Sigma^*)}{\mathcal{M}}$ such that $L = f \circ h$.
\end{definition}
Therefore, for $\omega$-categorical and $\omega$-stable structures, by Theorem~\ref{t:monoidal:closed}, the languages recognized by definable monoids are the same as the languages recognized by definable weighted automata.

\begin{theorem}[On languages recognized by linear monoids on definable bases]\label{t:monoid}
    Let $\struct{A}$ be an $\omega$-categorical and $\omega$-stable structure. The languages in $\classifying{A}$ recognised by definable weighted-automata are exactly the same as the languages recognised by linear monoids on definable bases.
\end{theorem}

\subsubsection{Probabilistic register machines}
In the classical setting of finite automata, probabilistic automata are a special kind of weighted automata --- there is just an additional requirement that the weights of the transitions of any state must be non-negative and sum up to $1$. This requirement does not translate directly to weighted automata in $\classifying{A}$ over definable sets for a single reason. If a set $X$ is not finite then there are some non-discrete probability measures on it. For example, if $N$ is the set of atoms in the basic Fraenkel-Mostowski model $\classifying{\struct{N}}$, then there is a measure $\mu$ that assigns to every finite set probability $0$ and to every cofinite set, probability $1$. Therefore, if the transition function assigns such a probability to a given state, then it violates the extra requirement in the original definition of a weighted automaton (Definition~\ref{d:weighted:automata}). Consequently, we have to either restrict to the discrete measures on a set or drop the extra requirement from the definition as we did in the preceding subsection. One can also think of the following definition as of a suitable generalisation of ultra-automaton from Definition~\ref{d:ultra:machine}. 

\begin{definition}[Probabilistic automaton]\label{d:probabilistic:machine}
    A probabilistic definable automaton consists of definable sets $S$ and $\Sigma$, called
the states and the alphabet, and the following data:
\begin{itemize}
    \item $p_0 \in \catw{m}(S)$ the initial probability on states $S$
    \item $p_F \in \catw{m}(S)$ the final probability on states $S$
    \item $\mor{\sigma}{\Sigma \times S}{\catw{m}(S)}$ the probabilistic transition relation
\end{itemize}
\end{definition}
Such an automaton assigns to every word $w \in\Sigma^*$ the probability that when starting in states $s_0$ the automaton reach states $s_F$ upon reading word $w$. By Theorem~\ref{t:definable:measure} $\catw{m}(S)$ is a convex linear combination of ultrafilters on $S$, therefore $\catw{m}(S)$ is a convex subset of the free vector space $F(\overline{S})$ and by Theorem~\ref{t:free:space} it can be treated as a convex subset of $\mathcal{R}^S$. In any case, the extension of $\sigma$ to $\mor{\overline{\sigma}}{\Sigma \times \catw{m}(S)}{\catw{m}(S)}$ is given by the formula (see Section~\ref{sec:probability}): $$\overline{\sigma}(a)(\mu)(S_0) = \sum_{r \in[0, 1]} r \cdot \mu(\{s \in S \colon \sigma(s)(a)(S_0) = r\})$$
This formula is linear in variable $\mu$, which runs over the basis of $F(\overline{S})$. Therefore, it induces a linear map: $\Sigma \times F(\overline{S}) \rightarrow F(\overline{S})$. This together with Theorem~\ref{t:stone:cech:compactification} yields the following characterisation of probabilistic automata.

\begin{theorem}[On probabilistic automata]\label{t:probabilistic:automaton}
    Let $\struct{A}$ be an $\omega$-categorical and $\omega$-stable structure. A probabilistic automaton in $\classifying{A}$ on a definable set $S$ is a special case of a definable weighted-automaton on the Stone–Čech compactification $\overline{S}$ of $S$.
\end{theorem}

\subsection{Organisation of the paper}
The rest of the paper contains the proofs and some additional details of the abovementioned theorems. The next section investigates the properties of ultrafilters on definable sets. The central theorem of this section is Theorem~\ref{t:stone:cech:compactification}. In Section~\ref{sec:definable:spaces} we study closure properties of vector spaces over definable basis. The main result is Theorem~\ref{t:monoidal:closed}, which is based on two technical lemmas: Lemma~\ref{l:independent} and Lemma~\ref{l:spans}. Section \ref{sec:probability} is devoted to studying probability measures on definable sets also known as Keisler measures to model theorists. The main result of the section is Theorem~\ref{t:definable:measure}. We conclude the paper in Section~\ref{sec:conclusions}. Appendix~\ref{sec:app:intro:theorem} gives the exact statement of a theorem mentioned in the introduction and supplies it with a proof. Appendix~\ref{sec:app:intro:theorem} contains a counterexample to the claim that in every space dual to the space on a definable basis has a definable basis -- Theorem~\ref{t:polar:counterexample} shows that it may not have any basis (definable or not) at all. Appendix~\ref{sec:app:proofs} contains some additional proofs of supplementary theorems, which are not crucial for the presented material.

\section{Ultrafilter monad}
\label{sec:ultrafilter:monad}
The aim of this section is to investigate ultrafilter monad on the category $\classifying{A}$ of symmetric sets over an $\omega$-categorical $\omega$-stable structure $\struct{A}$. First, let us observe that the ultrafilter monad exists on any Boolean topos, provided it satisfy some mild conditions about existence of free algebras. Moreover, such monad is always a \emph{strong} monad. Explicitly, every topos can be regarded as a category enriched over itself, i.e.~just put $\hom(A, B) = B^A$, where $B^A$ is the internal function space \cite{kelly1982basic}. In particular, when working in $\classifying{A}$ it is natural to think that $\hom(A, B)$ carries the group action. The same is true for other algebraic structures studied here, especially: vector spaces (modules) and Boolean algebras.  In fact, we have an enriched adjunction between the free vector space functor $\mor{F}{\classifying{A}}{\catw{Vect}_{\classifying{A}}}$ and the forgetful functor $\mor{S}{\catw{Vect}_{\classifying{A}}}{\classifying{A}}$. Similarly, we have an enriched adjunction between the free Boolean algebra functor $\classifying{A} \rightarrow \catw{Bool}_{\classifying{A}}$ and the underlying functor $\mor{|-|}{\catw{Bool}_{\classifying{A}}}{\classifying{A}}$ (these follow from the transfer principle from \cite{licsMRP} and the fact that both the theory of vector spaces and the theory of Boolean algebras are Lawvere theories.). Forgetful functors, being right adjoint, preserve all limits that exist, and free functors preserve all colimits that exist. Specifically, the enriched category of Boolean algebras have cotensors with all symmetric sets (these are just weighted limits), i.e.~for every symmetric set $A$ and an internal Boolean algebra $B$ the cotensor $B \pitchfork A$ exists and is preserved by the underlying functor $|B \pitchfork A| = |B| \pitchfork A = |B|^A$, where the last equality holds because cotensors coincide with exponents in the base of enrichment. Now, if we now consider the $2$-element Boolean algebra $2$, we have a series of enriched natural isomorphisms:
\begin{center}
\begin{tabular}{c}
$\hom_{\catw{Bool}_{\classifying{A}}^{op}}(2 \pitchfork X, B)$ \\
\hline\hline
$\hom_{\catw{Bool}_{\classifying{A}}}(B, 2 \pitchfork X)$\\
\hline\hline
$\hom_{\catw{Bool}_{\classifying{A}}}(B, 2)^X$\\
\hline\hline
$\hom_{\classifying{A}}(X, \hom_{\catw{Bool}_{\classifying{A}}}(B, 2))$\\
\end{tabular}
\end{center}
which means that $\mor{2 \pitchfork {(-)}}{\classifying{A}}{\catw{Bool}_{\classifying{A}}}$ is an internal left adjoint to enriched hom-functor $\mor{\hom_{\catw{Bool}_{\classifying{A}}}(-, 2)}{\catw{Bool}_{\classifying{A}}}{\classifying{A}}$. By composing these two functors we obtain a strong (internal, enriched) monad on $\classifying{A}$, i.e.~the ultrafilter monad: $\mor{\hom(2 \pitchfork (-), 2)}{\classifying{A}}{\classifying{A}}$, where we still write $2 \pitchfork (-)$ instead of $2^{(-)}$ to indicate that this operation is \emph{not} an exponent in Boolean algebras.

\begin{theorem}[Ultrafilter monad]\label{t:ultrafilter:strong}
    The ultrafilter monad on $\classifying{A}$ exists and is strong (equivalently, enriched over $\classifying{A}$).
\end{theorem}

As usual, we shall call algebras of the ultrafilter monad \emph{compact Hausdorff spaces}. We will also denote the monad by $\mor{(-)}{\classifying{A}}{\classifying{A}}$ to highlight the fact that the free algebra $\overline{X}$ on a given set $X$ is the ``free compactification'' of $X$, i.e.~the internal Stone–Čech compactification of $X$.

\begin{lemma}[On preservation of finite coproducts]\label{l:ultrafilter:finite:coproducts}
   Ultrafilter monad preserves binary coproducts. 
\end{lemma}
\begin{proof}
    The proof is pretty standard. Let $X$ and $Y$ be two symmetric sets and consider an symmetric ultrafilter $p$ on $X$. We shall define an ultrafilter $p^*$ on $X \sqcup Y$ as follows. For any $S \subset X \sqcup Y$ put $S \in p^*$ if and only if $S \cap X \in p$. Observe that $X \sqcup Y \setminus S \in p^*$ if and only if $(X \sqcup Y \setminus S) \cap X = X \setminus (S \cap X) \in p$, thus $(X \sqcup Y \setminus S \not\in p^*$. Similarly, if $S_1, S_2 \in p^*$ then $S_1 \cap X, S_2 \cap X \in p$, therefore  $(S_1 \cap S_2) \cap X \in p$, so  $S_1 \cap S_2 \in p^*$. Moreover, $p^*$ is obviously upward-closed. In the other direction, given an ultrafilter $q$ on $X \sqcup Y$ either $X \in q$ or $Y \in q$, but not both and we obtain an ultrafilter on $X$ (resp. $Y$) via restriction. It is also obvious that the operations are inverse of each other.
\end{proof}

Till the end of the section we shall assume that $\struct{A}$ is $\omega$-categorical, $\omega$-stable and that $\struct{A}$ eliminates imaginaries (extending a structure with elimination of imaginaries, as mentioned in the introduction, does not change the category $\classifying{A}$).

\begin{theorem}[Internal Stone–Čech compactification]\label{t:stone:cech:compactification}
Let $X$ be $A_0$-definable. The free Stone–Čech compactification $\overline{X}$ of $X$, i.e.~the set of ultrafilters on $X$, is $A_0$-definable.
\end{theorem}
\begin{proof}
Let us first assume $X = A^n$.
Let $\mor{\mu}{\mathcal{P}(X)}{2}$ be an $A_1$-supported ultrafilter on $X$. Consider any formula $\phi(x, \overline{y})$, where $\overline{y}$ are treated as parameters. Because $\mu$ is $A_1$-supported, the set: $D_\phi = \{\overline{q} \in A^{|\overline{y}|}\colon \mu(\phi(x, \overline{q})) = 1\}$ 
is $A_1$-supported. Therefore, by $\omega$-categoricity of $A$ set $D_\phi$ may be thought of as a formula $D_\phi(\overline{y}, \overline{a})$ with parameters $\overline{a}$. Therefore, the corresponding $\phi$-type is definable by $D_\phi(\overline{y}, \overline{a})$, so it is $A_1$-definable. Because, this is true for every formula $\phi(x, \overline{y})$, ultrafilter $\mor{\mu}{\mathcal{P}(X)}{2}$ corresponds to an $A_1$-definable type in $S_n(A)$. In the other direction, let us assume that a type $p \in S_n(A)$ is $A_1$-definable for some finite $A_1$. Then for every $\pi_{A_1}$ and every $\phi(x, \overline{q})$ we have that:
$\pi_{A_1}(\phi(x, \overline{q})) \in p \Leftrightarrow \phi(\pi_{A_1}(x), \overline{q})) \in p  \Leftrightarrow \phi(\pi_{A_1}(\overline{q}), \overline{a}) \Leftrightarrow  \phi(\overline{q}, \overline{a}) \Leftrightarrow \phi(x, \overline{q}) \in p$.
Thus, $p$ is an $A_1$-supported function. By Theorem~\ref{l:parameters:in:definable:types} space $S_n(A)$ has finitely many orbits. Therefore, $\overline{X}$ has finitely many orbits. 

Now, moving to the general case, observe that arbitrary $A_0$-definable set $X$ is an equivariant subset of $A^n$ for some finite $n$ in an expansion of structure $A$ by finitely many constants $A_0$. Because such an expansion preserves both $\omega$-categoricity and $\omega$-stability of a structure, without loss of generality we may assume that $X$ is an equivariant subset of $A^n$ in $\struct{A}$. Let us denote $X^c = A^n \setminus X$. By Lemma~\ref{l:ultrafilter:finite:coproducts} the ultrafilter monad preserves finite coproducts, therefore $\overline{A^n} = \overline{X \sqcup (A^n \setminus X)} = \overline{X} \sqcup \overline{A^n \setminus X}$. Because $\overline{A^n}$ has finitely many orbits, both $\overline{X}$ and $\overline{A^n \setminus X}$ must have finitely many orbits. 
\end{proof}


\begin{corollary}\label{c:definable:monad}
    The ultrafilter monad restricts to the monad on the full subcategory of $\classifying{A}$ of definable sets.
\end{corollary}

\begin{lemma}[Types in $\omega$-categorical $\omega$-stable structures]\label{l:parameters:in:definable:types}
    Let $\omega$-categorical $\omega$-stable structure and $p$ a type in $S_n(A)$. Then $p$ is supported by a finite tuple $\overline{a} \in A^k$ and definable by an $\overline{a}$-supported formula $\phi(x, \overline{a})$ of the same Morley rank as $p$ and Morley degree 1. Moreover $|\overline{a}| \leq 2|x|$.
\end{lemma}

The proof of the above lemma is strongly based on Theorem~6.3 from \cite{cherlin1985no}.

\begin{theorem}[Cherlin, Harrington, Lachlan \cite{cherlin1985no}]
    Let $\struct{A}$ be $\omega$-categorical and $\omega$-stable. Then any type $p \in S_1(A)$ is
definable by a normalised formula with two parameters.
\end{theorem}

 This is a bit stronger than the original statement from \cite{cherlin1985no}, but it can be extracted from the proof. The authors show that the defining formula $\phi(x, \overline{a})$ of type $p$ can be chosen to be normalised: i.e.~whenever $\overline{a}$ and $\overline{a'}$ are in the same orbit, and $\phi(x, \overline{a})$ differs from $\phi(x, \overline{a'})$ on a set of the Morley rank strictly smaller than $\phi(x, \overline{a})$ then in fact they differ on the empty set, i.e.: $\phi(x, \overline{a}) = \phi(x, \overline{a'})$. Moreover, the term ``defining formula'' refers to the following property: for every formula with parameters $\psi$ we have that $\psi(x, \overline{q}) \in p$ if and only if the Morley rank of $\phi(x, \overline{a}) \cap \psi(x, \overline{q})$ is the same as the Morley rank of $\phi(x, \overline{a})$, therefore it has the minimal Morley rank possible in $p$. It is also a standard result in stability theory (see: \cite{marker2006model}, \cite{hodges1993model} or \cite{tent2012course} for more details) that the above property is definable from parameters $\overline{a}$.

\begin{proof}[Proof of Lemma~\ref{l:parameters:in:definable:types}]
Let us recall that a type $p$ being definable from parameters $\overline{a}$ means just that for every formula $\psi(x, y)$ without parameters, there is a formula $D_{\psi}(y, \overline{a})$ with parameters $\overline{a}$ such that $\psi(x, \overline{q}) \in p$ if and only if $D_{\psi}(\overline{q}, \overline{a})$.  Therefore, definability by a single formula in the sense of \cite{cherlin1985no} is a stronger property. Fortunately, it is not stronger for $\omega$-categorical $\omega$-stable structures. We must show that parameters $\overline{a}$ are necessary, i.e.~that $p$ is not definable with a smaller number of parameters. But it follows from our choice of $\phi(x, \overline{a})$ to be normalised.
    The fact that $\phi(x, \overline{a})$ is normalised means that every permutation $\pi$ that fixes $p$ must also fix $\phi(x, \overline{a})$, otherwise we would have $\phi(x, \pi^{-1}(\overline{a})) \in p$ and therefore $\phi(x, \pi^{-1}(\overline{a})) \cap \phi(x, \overline{a}) \in p$. But because $\phi(x, \overline{a})$ was of the smallest Morley Rank, say $r$, then the Morley rank of $\phi(x, \pi^{-1}(\overline{a})) \cap \phi(x, \overline{a})$ must be equal to $r$ too. Therefore, $\phi(x, \overline{a}) \setminus \phi(x, \pi^{-1}(\overline{a})) \cap \phi(x, \overline{a})$ must be of Morley Rank strictly smaller than $r$ -- otherwise, $\phi(x, \overline{a})$ would be a disjoint sum of two sets of Morley Rank $r$, what would contradict the choice of $\phi(x, \overline{a})$ with Morley Degree $1$. The same argument works for $\phi(x, \pi^{-1}(\overline{a})) \setminus \phi(x, \pi^{-1}(\overline{a})) \cap \phi(x, \overline{a})$, what means that $\phi(x, \overline{a})$ and $\phi(x, \pi(\overline{a}))$ differs on a set of the Morley Rank strictly smaller than $r$. Therefore, must be equal.

    It remains to prove the bound on the number of parameters for types in $S_n(A)$. Note however that the general case for $n$-types reduces to the case of $1$-types. It suffices to observe that the reduct of an $\omega$-categorical $\omega$-stable structure is itself $\omega$-categorical and $\omega$-stable and one may easily construct a reduct of $A$ on $n$-element tuples of $A$ whose $1$-types encode $n$-types of $A$.
\end{proof}

\begin{example}[Ultrafilters in random graphs]\label{e:types:in:graphs}
    As mentioned in the introduction the structure $\struct{R}$ of the Random Graph from Example~\ref{e:graph} is $\omega$-categorical, but not $\omega$-stable. Here we will show, that the set of ultrafilters on $R$ has infinitely many orbits. First, observe that by the extension property of the Random Graph, every set of formulas:
    $$S = \{E(x, a_1), E(x, a_2), \dotsc, \neg E(x, b_1), \neg E(x, b_2), \dotsc \}$$
    for pairwise distinct elements $a_i, b_j$ is finitely satisfiable. Therefore, by the compactness of the First-Order Logic, it is satisfiable. Because, Random Graphs admit elimination of quantifiers, if $a_i, b_j$ enumerate the whole $R$, then $S$ generates an ultrafilter $\mu$ on $\mathcal{P}(R)$. Moreover, $\mu$ is symmetric if and only if the set $\{a_i \colon E(x, a_i) \in S\}$ is definable. Therefore, symmetric non-principal ultrafilters on $R$ are tantamount to definable subsets of $R$. 
\end{example}

\section{Closure properties of vector spaces on definable sets}\label{sec:definable:spaces}
The aim of this section is to prove that the category of vector spaces on definable bases enjoys many closure properties somehow similar and somehow different from the closure properties of the category of vector spaces on finite bases. Let $V, W$ be vector spaces on $A_0$-definable bases $\Lambda$ and $\Gamma$ respectively, in $\classifying{A}$ for an $\omega$-categorical and $\omega$-stable structure $\struct{A}$. Then:
\begin{itemize}
    \item the finite coproduct space $V + W$ is the same as the finite product space $V \times W$ and has basis $\lambda \sqcup \Gamma$
    \item the tensor product space $V \otimes W$ has basis $\lambda \times \Gamma$
    \item the dual space $V^* = \mathcal{K}^V$ has basis $\overline{\Lambda}$
    \item the space of linear exponent $V \multimap W$ has a basis that is an $A_0$-equivariant subset of $\overline{\Lambda \times \Gamma}$.
\end{itemize}
The last closure property is quite remarkable, because it means that the category of vector spaces on definable sets is monoidaly closed. This is in contrast to the category of definable sets, where the exponents are not definable. The first two properties on the above list follows from the same properties when we treat the vector spaces as living inside classical set theory, whilst the last property can be proved directly from the third one. Therefore, the main difficulty is in proving the characterisation of the basis of the dual space in terms of the basis of the space.

\begin{theorem}[Free space on ultrafilters]\label{t:free:space}
Let $X$ be definable. The free $\mathcal{K}$-vector space $F(\overline{X})$ over the set $\overline{X}$ of ultrafilters on $X$ is isomorphic to the space of functions $\mathcal{K}^X$.
\end{theorem}

Before we prove the theorem let us make an important remark.

\begin{remark}[On necessity of stability]\label{r:stability:is:necessary}
    Let us consider the set of atoms $Q$ in the ordered Fraenkel-Mostowski model $\classifying{Q}$. We claim that both $\overline{Q}$ and the basis $\Lambda$ of $\mathcal{R}^Q$ exist, but are \emph{not} isomorphic. In fact, $\Lambda$ is a proper subset of $\overline{Q}$. To see this,  consider a symmetric function $f \colon Q \rightarrow \mathbb{R}$. Then there is a finite decomposition $Q = I_1 \sqcup I_2 \sqcup \cdots I_n$ on intervals $I_k$, such that $f$ is constant on each $I_k$. Therefore, the set of vectors:
$\{1, p_1^*, p_2^*, \cdots, p_1^{<}, p_2^{<}, \cdots \} \approx Q \sqcup Q \sqcup 1$, where: $1$ is the constant function, i.e.~$1(a)=1$, $p^*$ is the characteristic function, i.e.~$p^*(q) = [p = q]$, $p^{<}$ is the open down set of $p$, i.e.~$p^{<}(q) = [p > q]$
generates $\mathbb{R}^Q$. Moreover, these vectors are linearly independent: if $S$ is any non-empty finite set of the above vectors, then there exists vector $m \in S$ and an atom $p$ with $m(p) = 1$ such that for every $s \in S \setminus \{m\}$ we have that $s(p) = 0$, so $m$ cannot be a linear combination of $S \setminus \{m\}$ and by induction on the size of $S$, the vectors are linearly independent.
On the other hand, it is easy to compute the Stone-Čech compactification $\overline{Q}$ of $Q$ directly:
$\overline{Q} = \{-\infty, +\infty, (q)_{q \in Q}, (q^{-})_{q \in Q}, (q^{+})_{q \in Q}\} \approx Q \sqcup Q \sqcup Q \sqcup 2$
where:
    $-\infty$ is the ultrafilter generated by $\{x \colon x < q\}_{q \in Q}$,
    $+\infty$ is the ultrafilter generated by $\{x \colon x > q\}_{q \in Q}$,
    $(q)_{q \in Q}$ are all principal ultrafilters,
    $(q^{-})_{q \in Q}$ are ultrafilters of generated by the left neighbourhoods of $q$, i.e.~all sets $\{x \colon p < x < q\}_{p < q}$,
    $(q^{-})_{q \in Q}$ are ultrafilters of generated by the right neighbourhoods of $q$, i.e.~all sets $\{x \colon q < x < p\}_{p > q}$.
\end{remark}

The proof of Theorem~\ref{t:free:space} is contained in Lemma~\ref{l:independent} and Lemma~\ref{l:spans}, but before we state the lemmas we have to fix some terminology.
For a set $X$ let us denote by $U$ the set of non-principal ultrafilters on $X$, i.e.~$U(X) = \overline{X} \setminus X$, where $X$ is identified with the image of $X$ under $\eta$. With every subset $Y \subseteq \overline{X}$ we may associate the set $Y'$ of limit points of $Y$, i.e.~the image of $U(Y) \subseteq U(\overline{X}) \subseteq \overline{\overline{X}}\to^\mu \overline{X}$. It is clear that if $Y$ is closed under limits, then $Y' \subseteq Y$ and $Y'$ is also closed under limits (i.e.~it is defined as the subspace of limit points). Moreover, if $Y$ is $A_0$-supported then $Y'$ is $A_0$-supported, because both $\eta$ and $\mu$ are equivariant. Therefore, the operation of taking limit points produces a descending sequence of $A_0$-equivariant subsets of $Y$:
$$\cdots \subseteq Y^{(n+1)} \subseteq Y^{(n)} \subseteq \cdots \subseteq Y'' \subseteq Y' \subseteq Y$$
Because $A$ is $\omega$-categorical, there are only finitely many $A_0$-supported subsets of $Y$, so there must be $n$ such that $(Y^{(n)})' = Y^{(n)}$. Therefore, the sequence gives a decomposition of $Y$ on $n$-disjoint $A_0$-equivariant subsets $Y_{(i)} = Y^{(i-1)} \setminus Y^{(i)}$. Observe also, that by the construction, points in $Y_{(i)}$ are isolated in $Y^{(i-1)}$. We have the following lemma.

\begin{lemma}[Existence of a binary tree]\label{l:binary:tree}
If the sequence $Y^{(n)} \subseteq \cdots \subseteq Y'' \subseteq Y' \subseteq Y$ ends in a non-empty set $Y^{(n)}$ then structure $A$ is not $\omega$-stable.
\end{lemma}
\begin{proof}
Let us write $P = Y^{(n)}$. Observe that if $P$ is finite then $P' = \emptyset$ so $P = \emptyset$. If $P$ is infinite, then we may choose any two distinct points $a \neq b \in P$ and two non-principal ultrafilters $p, q \in U(P)$ such that $p \rightarrow a$ and $q \rightarrow b$. Obviously, $p \neq q$, so there must be a set $P_0$ such that $P_0 \in p$ and $P_0 \not\in q$. This means that $P_1 = (P \setminus P_0) \in q$. Because, $p$ and $q$ are non-principal, both $P_1$ and $P_2$ are infinite. Therefore, $P_1$ and $P_2$ give a decomposition of $P$ on two disjoint infinite subsets and because $U(P) = U(P_1 \sqcup P_2) = U(P_1) \sqcup U(P_2)$, we may construct by induction an infinite binary tree $(P_w)_{w \in \{0, 1\}^*}$ what contradicts $\omega$-stability of $A$.    
\end{proof}

\begin{remark}[Non-principal ultrafilters on infinite sets]\label{r:non:principal:ultrafilters:exist}
The proof of Lemma~\ref{l:binary:tree} shows that if $U(X)$ is empty, i.e.~there are no non-principal ultrafilters on $U(X)$, then $X$ must be finite. Therefore, for every infinite definable set $X$ the set of non-principal ultrafilters on $X$ is non-empty.     
\end{remark}

For the rest of the proof, we shall assume that $A$ is $\omega$-stable (therefore, $Y^{(i-1)} = \emptyset$) and without loss of generality that $Y = \overline{X}$ is equivariant.

\begin{remark}
Because $A$ is $\omega$-stable, every type in $S_n(A)$ is definable with a finite set of parameters. Therefore, by the considerations in the proof of Theorem~\ref{t:stone:cech:compactification}, $S_n(A)$ is isomorphic to $\overline{A^n}$. Moreover, a formula $\phi$ (with parameters from $A$) has Morley rank $r$ and Morley degree $d$ if and only if it belongs to exactly $d$ types in $\overline{A^n}_{(r)}$. According to this setting, Lemma~\ref{l:binary:tree} gives an internal proof of the fact that formulas in an $\omega$-categorical $\omega$-stable theory must have finite Morley rank.      
\end{remark}

\begin{lemma}[On nice isolated sets]\label{l:nice:exists}
For every $y \in Y_{(i)}$ for $1 \leq i \leq n$ there is a set $\lambda_y$ that isolates $y$ in $Y^{(i-1)}$ and such that if $A_0$ supports $y$ then $A_0$ supports $\lambda_y$. Moreover, we may choose $\lambda_y$ uniformly for the orbit of $y$, i.e.~$\pi(\lambda_y) = \lambda_{\pi(y)}$.
\end{lemma}
\begin{proof}
Set $\lambda_y$ may be chosen to be the normalised defining formula for $y$ from Lemma~\ref{l:parameters:in:definable:types}. Consider any permutation $\pi$. If $\lambda_y$ is normalised, then $\pi(\lambda_y)$ must be normalised (by the definition of normality) and belongs to $\pi(y)$. Because both the rank and the degree are preserved by automorphisms, $\pi(\lambda_y)$ can be chosen as a defining formula for $\pi(y)$.
\end{proof}

The above lemma implies the existence of an equivariant injection $\overline{X} \rightarrow \mathcal{P}(X)$ sending an ultrafilter $p$ to a nice set contained in $p$.

\begin{lemma}[Nice isolated sets are linearly independent in $\mathcal{K}^X$]\label{l:independent}
The proof is by induction on sets $\overline{X}_{(i)}$. For $i = 1$ sets $\lambda_x$ are singletons $\{x\}$, so they are linearly independent as functions $X \rightarrow \{0, 1\} \rightarrow \mathcal{K}$ for any ring $\mathcal{K}$. Let us assume that the set $\{\lambda_p \colon \exists_{k < i} \; p \in \overline{X}_{(k)}\}$ is linearly independent. Consider any linear combination that equals $0$, that is: 
$\alpha + a_1 \lambda_{p_1} + a_2 \lambda_{p_2} + \cdots a_n \lambda_{p_n} = 0$, where $p_1, p_2, \dotsc,  p_n$ belong to $\overline{X}_{(i)}$ and $\alpha$ is a linear combination of some $\lambda_p$ for $p \in \overline{X}_{(k)}$ for $k < i$. Because $\alpha$ is a finite combination of functions that are zero at every $p \in \overline{X}_{(i)}$, function $\alpha$ must be itself zero at every $p \in \overline{X}_{(i)}$, therefore $a_1 \lambda_{p_1} + a_2 \lambda_{p_2} + \cdots a_n \lambda_{p_n}$ must be zero at every $p \in \overline{X}_{(i)}$. But $(a_1 \lambda_{p_1} + a_2 \lambda_{p_2} + \cdots a_n \lambda_{p_n})(p) = a_j$ for $p = p_j$ and so $a_1 = a_2 = \cdots = a_n = 0$. And then, $\alpha = 0$.
\end{lemma}

\begin{lemma}[Nice isolated sets spans $\mathcal{K}^X$]\label{l:spans}
Let $\mor{f}{X}{\mathcal{K}}$ be a finitely supported function to a classical (i.e.~without atoms) ring $\mathcal{K}$. Because $\mathcal{K}$ is classical and $X$ has finitely many orbits, such $f$ must take only finitely many values in $\mathcal{K}$, say $r_1, r_2, \dotsc, r_n$. These values induce decomposition of $X$ into $n$ disjoint subsets $\phi_i = f^{-1}[r_i] \subseteq X$, such that $f = \sum_{i=1}^n r_i\phi_i$, where $\phi_i$ are treated as characteristic functions $\mor{\phi_i}{X}{\{0, 1\} \rightarrow \mathcal{K}}$. Moreover, we may drop $i$ such that $r_i = 0$ from the sum. In the below we shall restrict to $\phi_j$ such that $r_j \neq 0$. 

The proof is by induction on $k$ such that $k = \max_{0 \leq i < n} p \in \overline{X}_{(i)} \land \phi_j \in p$, i.e.~the biggest $i$ such that $\phi_j$ belongs to an ultrafilter in  $\overline{X}_{(i)}$. For $k=0$ subset $\phi_j$ must be finite, therefore it is the sum of singletons $\{x\}$ such that $x \in \phi_j$ and the sum is disjoint, thus interpreted the same way in any ring $\mathcal{K}$. Let us now assume that the theorem is true for all $k' < k$. By the assumption $\phi_j$ does not belong to an ultrafilter from $\overline{X}_{(k'')}$ for $k'' > k$, so the number of ultrafilters $p$ from $\overline{X}_{(k)}$ that $\phi_j$ belongs to $p$ must be finite, say the ultrafilters are $p_1, p_2, \dotsc, p_m$. Then $g = \phi_j - (\lambda{p_1} + \lambda_{p_2} + \cdots + \lambda_{p_s})$ is a finitely supported function such that none of $g^{-1}[r]$ for $r \neq 0$ is contained in an ultrafilter from  $\overline{X}_{(k'')}$ for $k'' > k-1$. By inductive hypothesis $g$ is a linear combination of nice isolated sets, say $a_1\lambda{p_1'} + a_2\lambda_{p_2'} + \cdots + a_t\lambda_{p_t'}$, therefore $\phi_j = a_1\lambda{p_1'} + a_2\lambda_{p_2'} + \cdots + a_t\lambda_{p_t'} + \lambda{p_1} + \lambda_{p_2} + \cdots + \lambda_{p_s}$.     

\end{lemma}

\begin{remark}[A few words about dual modules]
    Although this section is devoted to vector spaces, that is: modules over a field, an inspection of Lemma~\ref{l:independent} and Lemma~\ref{l:spans} shows that Theorem~\ref{t:free:space} holds for modules over arbitrary ring. 
\end{remark}

Theorem~\ref{t:free:space} together with Lemma~\ref{l:parameters:in:definable:types} say that the vector space of linear functionals on an $A_0$-definable set has an $A_0$-definable basis. We can slightly extend the theorem to include sets of bounded support. We start with a lemma, which says that linear functionals on a definable set cotensored with any classical set has a basis.

\begin{lemma}[On dual basis of $\kappa \times X$]\label{l:orbit:cardinal:space}
Let $X$ be an equivariant set consisting of a single orbit and $\kappa$ a (classical) cardinal. Then the vector space $\mathcal{K}^{\kappa \times X}$ has an equivariant basis $\mathcal{B}(\kappa) \times \overline{X}$, where $\mathcal{B}(\kappa)$ is a basis of $\mathcal{K}^\kappa$.
\end{lemma}

Note, however, that it is no longer true that the basis of $\mathcal{K}^{\kappa \times X}$ can be obtained as the Stone-Čech compactification of $\kappa \times X$ -- the later is just much bigger for infinite cardinals $\kappa$. The proof of Lemma~\ref{l:orbit:cardinal:space} and next Theorem~\ref{l:orbit:cardinal:space} is in Appendix~\ref{sec:app:proofs}.

\begin{theorem}[On the existence of dual basis]\label{t:dual:basis}
For every $\omega$-categorical $\omega$-stable structure $A$, the set theory with atoms over $A$ satisfies the following: For every $X$ of bounded support the vector space $\mathcal{K}^X$ has a basis of a bounded support for every classical field $\mathcal{K}$. Moreover, if $X$ is $A_0$-equivariant (resp.~$A_0$-definable), then we may choose the basis to be $A_0$-equivariant (resp.~$A_0$-definable).
\end{theorem}

Nonetheless, we do not know if every vector space of the form $\mathcal{K}^X$ (where $X$ is not necessarily definable $X$) in $\classifying{A}$ for $\omega$-categorical $\omega$-stable structure has a basis. The following problem is crucial for answering this question.

\begin{problem}[Dual basis of $\mathcal{P}_\mathit{fin}(A)$]
Let $A$ be the set of atoms in $\classifying{A}$ for $\omega$-categorical $\omega$-stable structure $\struct{A}$ and denote by $\mathcal{P}_\mathit{fin}(A)$ the set of finite (finitely supported) subsets of $A$. Does the vector space $\mathcal{K}^{\mathcal{P}_\mathit{fin}(A)}$ have a basis?
\end{problem}
We do not know the answer even in case $\mathcal{K} = 2$ and $\struct{A}$ is the stucture of pure sets from Example~\ref{e:pure:sets}.

We close these considerations by proving that for vector spaces $V$ and $W$ with $A_0$-definable basis the space of linear functions $V \multimap W$ has an $A_0$-definable basis.

\begin{theorem}[Vector spaces over definable sets are monoidaly closed]\label{t:monoidal:closed}
Let $X$ and $Y$ be $A_0$-definable sets in $\classifying{A}$ for $\omega$-categorical $\omega$-stable structure $\struct{A}$. The space of linear functions $\mathit{Lin}(F(X), F(Y))$ from $F(X)$ to $F(Y)$ has an $A_0$-definable basis.    
\end{theorem}
\begin{proof}
    It suffices to prove the theorem for equivariant $X, Y$. Because $F(X)$ is free, there is an isomorphism $\mathit{Lin}(F(X), F(Y)) \approx F(Y)^X$. Because $\overline{Y} = U(Y) \sqcup Y$ we have that: $F(\overline{Y})^X \approx F(U(Y) \sqcup Y)^X \approx F(U(Y))^X \times F(Y)^X$. Therefore, $F(Y)^X$ is a closed subspace of $F(\overline{Y})^X$. On the other hand, $F(\overline{Y}) \approx \mathcal{K}^Y$ and so $F(\overline{Y})^X \approx (\mathcal{K}^Y)^X \approx \mathcal{K}^{X \times Y}$. By Theorem~\ref{t:free:space} vector space $\mathcal{K}^{X \times Y}$ has a basis isomorphic to $\overline{X \times Y}$. Explicitly, the basis consists of normalised sets $\lambda_p$ for each type $p$ of $X \times Y$. Notice, however, that any linear combination of these $\mor{\lambda_p}{X}{F(U(Y)) \times F(Y)}$ that have a non-zero component in $F(U(Y))$ must have a non-zero component in $F(U(Y))$, because otherwise the set corresponding to $F(U(Y))$ would be finite contradicting the definition of $U(Y)$. Therefore, the set of these $\lambda_p$ that factors through $F(Y)$ form a basis of $\mathit{Lin}(F(X), F(Y))$.
\end{proof}
For example, the basis of $\mathit{Lin}(F(A), F(A))$ in $\classifying{\mathcal{N}}$ consists of: the identity $\lambda x . x$; the constant functions $\lambda x . a$ for every $a \in A$; the ``singletons'', i.e.~functions that map $a \mapsto b$ for fixed $a, b \in A$ and all other elements $x \neq a$ to $0$.

\subsection{A few notes on unstable theories}
\label{ss:unstable}
We saw in Example~\ref{e:types:in:graphs} that for an unstable theory, the set of (definable) ultrafilters on a definable set need not be definable. Moreover, we saw in Remark~\ref{r:stability:is:necessary} that for an unstable theory even if the set of ultrafilters is definable, it does not have to correspond to a basis of the dual space. These two observations bring at least three questions.
\begin{enumerate}
  \item Does every dual vector space $\mathcal{K}^X$ for $X$ definable in a not necessarily $\omega$-stable theory have a basis? 
\end{enumerate}
The answer is no by Theorem~\ref{t:polar:counterexample} from Appendix~\ref{sec:app:polar}. The theorem shows that this property may not hold even in case of very well-behaved theories. In particular, for every prime number $p$ we may construct an $\omega$-categorical, ultrahomogenous, simple structure $\struct{G}_{\mathcal{F}_p}$ such that $\mathcal{F}_p^{X}$ does not have a basis in $\classifying{\struct{G}_{\mathcal{F}_p}}$ for definable $X$. 
Nonetheless, in Remark~\ref{r:stability:is:necessary} we show that $\mathbf{K}^Q$ has a definable basis in $\classifying{\struct{Q}}$ for unstable structure $\struct{Q}$ of rational numbers with their natural ordering (DLO). But what about $\mathbf{K}^X$ for other definable sets in $\classifying{\struct{Q}}$?
\begin{enumerate}
    \setcounter{enumi}{1}
    \item Does every dual vector space $\mathcal{K}^X$ for $X$ definable in DLO have a basis?
\end{enumerate}
Theorem~\ref{t:rational:basis} shows that for every definable $X$ in DLO the dual vector space $\mathcal{K}^X$ has a definable basis. Note however, that Theorem~\ref{t:definable:types:dlo} tells us that for every definable set $X$ in DLO, its Stone-Čech compactification is also definable. Therefore, one may wonder the following. 
\begin{enumerate}
    \setcounter{enumi}{2}
    \item Does every dual vector space $\mathcal{K}^X$ for $X$ definable in a not necessarily $\omega$-stable theory have a basis on condition $\overline{X}$ is definable?
\end{enumerate}

We do not know the answer to the question, but we suspects that the answer is affirmative.
\begin{conjecture}
    Let $\struct{A}$ be an $\omega$-categorical structure. The following are equivalent:
    \begin{itemize}
        \item $\struct{A}$ is NIP
        \item for every definable set $X$ the vector space $\mathbf{K}^X$ has a definable basis in $\classifying{\struct{A}}$
        \item for every definable set $X$ its Stone-Čech compactification $\overline{X}$ is definable in $\classifying{\struct{A}}$
    \end{itemize}
\end{conjecture}
For the definition and basic properties of NIP theories see: \cite{simon2015guide}.
\subsubsection{On dense linear orderings}
Let us recall that theory DLO from \ref{e:rationals} has quantifier elimination, therefore every formula is a finite disjoint disjunction of conjunctions of atomic formulas. The atomic formulas are of the form $x < y$ or $x = y$, where $x, y$ can be either variables or parameters from $\mathcal{Q}$. By the above, a single $A_0$-orbit is a conjunction of formulas of the form $x_i < q$, $x_i > q$ or $x_i = q$ for $q \in A_0$ -- i.e.~the set defined by a single-orbit formula is a hyperrectangle with possibly infinite sides restricted to the half-hyperspace $x_1 < x_2 < \dotsc < x_n$. We would like associate with formulas of DLO an invariant like we did for stable theories, but the usual construction would not work here (i.e.~every infinite formula in DLO has an infinite Morley rank). In fact, there is no general theory of dimension for unstable theories. In the particular case of DLO, one could develop the notion of dimension through the machinery of Thorn-forking, but for our applications it suffices to define an ad-hoc notion of dimension in the following way.

Let us consider a slightly bigger model of DLO than $Q$, namely the set of real numbers $\mathcal{R}$ together with their natural ordering. Then we say that a formula $\phi$ with parameters from $Q$ is $n$-dimensional if and only if the set $\phi(R)$ has a non zero $n$-dimensional Lebesgue measure and is of measure zero according to $n+1$-dimensional Lebesgue measure. From the definition, we have that if $\phi$ is $n$ dimensional, than it cannot be a union of finitely (even countably!) many formulas of dimension $k < n$. We shall write $\mathit{dim}(\phi)$ for the dimension of $\phi$.
 
Fix $n$. Every weakly increasing sequence $q_1 \leq q_2 \leq \dotsc \leq q_n$ of $n$ rational numbers extended with $\infty$ defines an infinite $n$-dimensional hyperrectangle:
$$H(\overline{q}) = \{\tuple{x_1, x_2, \cdots, x_n} \colon x_1 < q_1 \land x_2 < q_2 \land \dotsc \land x_n < q_n \}$$
Moreover, for every choice $C$ of $n-k$ variables we have a $k$-dimensional hyperrectangle defined as:
$$H^C(\overline{q}) = \{\tuple{x_1, x_2, \cdots, x_n} \colon \forall_{x_i \in C}\; x_i < q_i \land \forall_{x_i \not\in C} x_i = q_i\}$$
We shall write $T^C(\overline{q})$ for the $k$-dimensional truncated hyperrectangle $H^C(\overline{q}) \cap A^{<n}$. Observe that the dimensions of hyperrectangles agree with the dimensions of their defining formulas. 
Let us denote by $Q^{<n}$ the set $\{\tuple{x_1, x_2, \dotsc, x_n} \in Q^n \colon  x_1 < x_2 < \dotsc < x_n\}$.
\begin{remark}\label{r:intersect:dim}
    If $T^C(\overline{q})$ and $T^{C'}(\overline{q'})$ are $k$-dimensional and $C \neq C'$ then the dimension of $T^C(\overline{q}) \cap T^{C'}(\overline{q'})$ is strictly smaller than $k$. This is because if $C \neq C'$ then there must be $x_i \in C' \land x_i \not\in C$ and then $T^C(\overline{q}) \cap T^{C'}(\overline{q'}) \subseteq T^{C \cup \{x_i\}}(\overline{s})$, where $\overline{s}$ is just $\overline{q}$ with $q_i$ substituted with $q'_i$. 
\end{remark}

\begin{lemma}[Truncated hyperrectangles are linearly independent]\label{l:hyper:independent}
The sets $T^C(\overline{q})$ are linearly independent in $\mathcal{K}^{Q^{<n}}$.
\end{lemma}
\begin{proof}
The proof proceeds by induction over the dimension $k$ of truncated hyperrectangles and sequences $\overline{q}$ with the lexicographical order. For $k=0$ the lemma is obvious. So let us assume $k > 0$. Consider any $T^C(\overline{q})$ of dimension $k$. Denote by $S$ the space spanned by all $T^C(\overline{q'})$ for $\overline{q'}$ strictly smaller than $\overline{q}$ in the lexicographical order, i.e. $S = \mathit{span}(\{T^C(\overline{q'}) \colon \overline{q'} < \overline{q}\})$. Denote by $V$ the space spanned by all $k-1$-dimensional truncated hyperrectangles together with $k$-dimensional truncated hyperrectangles with $T^{C_i}(\overline{q'})$ for $C_i \neq C$. We claim that $T^C(\overline{q}) \not\in S \oplus V$. For contradiction, suppose that $T^C(\overline{q}) \in S \oplus V$, what means that there are some tuples $\overline{q^1} < \overline{q^2} < \dotsc < \overline{q^k} < \overline{q}$ such that $T^C(\overline{q}) = \sum_{i=1}^a r_i T^C(\overline{q^i}) + \sum_{i=1}^b s_i T^{D_i}(\overline{{q'}^i})$, where $\mathit{dim}(T^{D_i}(\overline{{q'}^i})) \leq k$ and $D_i \neq C$. On the other hand the set $P = \{\overline{p} \colon \forall_{x_i \in C}\; q^k_i < p_i < q_i \land \forall_{x_i \not\in C}\; p_i = q_i\} \subseteq T^C(\overline{q})$ has dimension $k$ (i.e.~has a non-zero $k$-dimensional Lebesgue measure). Therefore, $T^C(\overline{q})$ and $\bigcup_{i=1}^a T^{C}(\overline{q_i})$ differ on a set of dimension $k$. Moreover, by Remark~\ref{r:intersect:dim} above, the intersection:
$$I = T^C(\overline{q}) \cap \bigcup_{i=1}^b T^{D_i}(\overline{{q'}^i}) =  \bigcup_{i=1}^b T^C(\overline{q}) \cap T^{D_i}(\overline{{q'}^i})$$
has dimension strictly smaller than $k$. Therefore $P \setminus I$ is non-empty and picking any $\overline{p} \in P \setminus I$ leads to the contradiction:     
    $$1 = T^C(\overline{q})(\overline{p}) = \sum_{i=1}^k r_i T^C(\overline{q_i})(\overline{p}) + \sum_{i=1}^b s_i T^{D_i}(\overline{{q'}^i})(\overline{p}) = 0$$
what completes the inductive step.
\end{proof}

\begin{lemma}[Truncated hyperrectangles span $\mathcal{K}^{Q^{<n}}$]\label{l:hyper:spans}
The sets $T_k(\overline{q})$ span $\mathcal{K}^{Q^{<n}}$.
\end{lemma}
\begin{proof}
Let $\phi \subseteq A^{<n}$ be $A_0$-supported. By $\omega$-categoricity it can be written as a disjoint union of its $A_0$-orbits. Therefore, it suffices to show that every $A_0$-supported orbit can be obtained as a linear combination of $T^C(\overline{q})$. By quantifier elimination, a single orbit is a conjunction of formulas of the form $x_i < q$, $x_i > q$ or $x_i = q$ for $q \in A_0$ -- i.e.~the set defined by a single-orbit formula is a hyperrectangle with possibly infinite sides restricted to the half-hyperspace $x_1 < x_2 < \dotsc < x_n$. The fact that we can obtain any hyperrectangle from hyperrectangles of the form $H^C(\overline{q})$ for arbitrary $\overline{q}$ (i.e.~not necessarily weakly increasing) is classic, but the exact formula is clumsy and depends on the characteristic of the field $\mathcal{K}$. First observe that we can ``flip'' any $H^C(\overline{q})$ by replacing a constraint $x_i < q_i$ with $x_i \geq q_i$ for any $x_i \in C$ in the following way: if we substitute $q_i$ in $\overline{q}$ with $\infty$ to obtain $\overline{q'}$ then because $H^C(\overline{q'}) - H^C(\overline{q}) = H^C(\overline{q'}) \cap \neg H^C(\overline{q})$ and $\neg H^C(\overline{q}) = \{\overline{x} \colon \exists_{x_j \in C}\;  x_j \geq q_j \lor \exists_{x_j \not\in C}\; x_j \neq q_j\}$ we have that:
$$H^C(\overline{q'}) - H^C(\overline{q})  = \{\overline{x} \colon \forall_{x_j \in C \setminus \{x_i\}}\; x_j < q_j \land x_i \geq q_i \land \forall_{x_j \not\in C}\; x_j = q_j\}\}$$   
For a $k$-dimensional hyperrectangle $H$ we have to add/subtract $H^C(\overline{q})$ for all $\overline{q}$ located at the corners of $H$ and then supply them with possibly missing $k-1$-dimensional faces. Moreover, $X \cap x_1 < x_2 < \dotsc < x_n$ is defined by the same hyperrectangles but intersected with $x_1 < x_2 < \dotsc < x_n$, what completes the proof.    
\end{proof}

\begin{theorem}[Dual basis in DLO]\label{t:rational:basis}
    Let $X$ be a definable set in the Ordered Fraenkel-Mostowski Model of Set Theory with Atoms. Then for any field $\mathcal{K}$ the vector space $\mathcal{K}^X$ has a definable basis.
\end{theorem}
\begin{proof}
    For simplicity of the proof we shall assume that $X$ is equivariant. The general case is analogous.
    By $\omega$-categoricity of $\struct{Q}$ $X$ is a finite disjoint union of its orbits $X = \bigsqcup_{i=1}^n X_i$ where $X_i \approx Q^{<n_i}$ for some $n_i$. By Lemma~\ref{l:hyper:independent} and Lemma~\ref{l:hyper:spans} we have that $\mathcal{K}^{Q^{<n_i}}$ has an equivariant definable basis $\Lambda_i$ consisting of truncated hyperrectangles. Therefore:
    $$\mathcal{K}^{X} \approx \mathcal{K}^{\bigsqcup_{i=1}^n X_i} \approx \prod_{i=1}^n \mathcal{K}^{X_i} \approx \prod_{i=1}^n F(\Lambda_i) \approx F(\bigsqcup_{i=1}^n \Lambda_i)$$
    Therefore, $\bigsqcup_{i=1}^n \Lambda_i$ is a basis of $\mathcal{K}^X$.
\end{proof}


\begin{theorem}[Stone-Čech compactification in DLO]\label{t:definable:types:dlo}
Let $X$ be a definable set in the Ordered Fraenkel-Mostowski Model of Set Theory with Atoms. Then its Stone-Čech compactification $\overline{X}$ is definable.
\end{theorem}
For simplicity of the proof we shall assume that $X$ is equivariant. The general case is analogous.
    By Lemma~\ref{l:ultrafilter:finite:coproducts} it is sufficient to prove the claim for sets of the form $Q^n$ and by remarks in proof of Theorem~\ref{t:stone:cech:compactification} it is sufficient to consider finitely definable types $S^{fd}_n(Q)$ in $S_n(Q)$ definable in a finite number of parameters $Q_0$. The proof proceeds by induction on the number $n$ of variables. For $n=1$ the types are described in Remark~\ref{r:stability:is:necessary}. Consider an $n$-type $p$. There are two cases: either $p$ contains a formula $x_i = x_j$ for distinct variables $x_i, x_j$, or not. In the first case, $p$ is completely determined by an $n-1$-type, so let us focus on the second case. We have that $x_i \neq x_j$ for every pair of distinct variables $x_i, x_j$. Then it must be the case that either $x_i < x_j$ or $x_j < x_i$ for every distinct pair of variables $x_i, x_j$ since $p$ is a type. Up to a permutation of variables, we may assume that the formula is of the form $x_1 < x_2 < \dotsc < x_n$. Let $\phi \in p$ and assume that $\phi$ is supported by $A_0 \subset Q$. Then $\phi$ is a finite disjoint union of its $A_0$-orbits. Because $p$ is a type, we may assume that one of its orbits belongs to $p$. This means, that such single-orbit formulas generate $p$. Therefore, we shall restrict to such formulas only. For fixed $A_0$ they are just conjunctions of: $x_i < a_i$ and $x_i > b_i$ for $a_i, b_i \in A_0$.  We claim that type $p$ is determined by two types $s \in S^{fd}_1(Q)$ and $q \in S^{fd}_{n-1}(Q)$. The following lemma says a bit more.

\begin{lemma}[Decomposition of types in DLO]\label{l:decomposition:dlo}
    If $\overline{x}$ is a sequence of $n$ variables then for every proper subsequence $x_{i_1}, x_{i_2}, \dotsc, x_{i_k}$ we have a projection map $\mor{\pi_k}{S^{fd}_n(Q)}{S^{fd}_k(Q)}$ defined as follows:
    $$\pi(p) = \{ \exists_{x_{i_{k+1}}, x_{i_{k+2}}, \dotsc, x_{i_n}} \phi(x_1, x_2, \dotsc, x_n) \colon \phi(x_1, x_2, \dotsc, x_n) \in p\} $$
    Let us denote by $S_n^{<}(Q) \subseteq S^{fd}_n(Q)$ the set of types $p$ such that: $x_1 < x_2 < \cdots < x_n \in p$. Then the mapping $\mor{\nabla}{S_n^{<}(Q)}{S^{fd}_1(Q) \times S_{n-1}^{<}(Q)}$ defined as: $\nabla(p) = \tuple{\pi_1(p), \pi_{n-1}(p)}$
    is injective.
\end{lemma}
\begin{proof}
Let $p \neq q \in S_n^{<}(Q)$, then by the above observation there must be a formula $\phi(x_1, x_2, \dotsc, x_n)$ that is a conjunction of: $x_i < a_i$ and $x_i > b_i$ for $a_i, b_i \in A_0$ and some $A_0$.
such that $\phi(x_1, x_2, \dotsc, x_n) \in p$ and $\phi(x_1, x_2, \dotsc, x_n) \not\in q$. Because $\phi(x_1, x_2, \dotsc, x_n) \equiv \psi(x_1) \land \psi'(x_2, \dotsc, x_n) \land x_1 < x_2$ and $x_1 < x_2$ belongs to both types, it must be that either $\psi(x_1) \not\in q$ and then $\psi(x_1) \not\in \pi_1(q)$ or $\psi'(x_2, \dotsc, x_n) \not\in q$ and then $\psi'(x_2, \dotsc, x_n) \not\in \pi_{n-1}(q)$.
\end{proof}
Now we can finish the proof of Theorem~\ref{t:definable:types:dlo}. By the inductive hypothesis we can assume that for $k < n$ sets $S^{fd}_k(Q)$ are definable equivariant sets. Every type $p \in S^{fd}_n(Q)$ contains exactly one formula $\phi \in S_n(\emptyset)$. By $\omega$-categoricity of $\struct{Q}$ the set $S_n(\emptyset)$ is finite. Therefore, $S^{fd}_n(Q)$ decomposes on finitely many sets $S_n^\phi(Q)$ where $\phi \in S_n(\emptyset)$. If $\phi$ includes equality between variables $x_i = x_j$, then $S_n^\phi(Q)$ is an equivariant subset of $S^{fd}_k(Q)$ for $k<n$ and by inductive hypothesis is equivariant definable; otherwise, i.e.~if $\phi$ does not include equality, then $S_n^\phi(Q)= S_n^{<}(Q)$ and by Lemma~\ref{l:decomposition:dlo}, there is an equivariant injection $S_n^{<}(Q) \rightarrow S^{fd}_1(Q)^n$, and so  $S_n^{<}(Q)$ is a equivariant definable. This means that $S_n^\phi(Q)$ is equivariant definable as it is a finite union of equivarian definable sets $S_n^\phi(Q)$.

\section{Probability measures}\label{sec:probability}
By a measurable space we shall mean a tuple $\tuple{X, \sigma}$, where $X$ is a set, and
$\sigma \subseteq \mathcal{P}(X)$ is a subset of the power-set of $X$ closed under Boolean operations
and countable unions/intersections:
\begin{itemize}
    \item $\emptyset \in \sigma$
    \item if $X_0 \in \sigma$ then $X \setminus X_0 \in \sigma$
    \item  if $(X_i)_{i\in \mathcal{N}}$ is a family of $X_i \in \sigma$ then $(\bigcup_{i\in \mathcal{N}} X_i) \in \sigma$
\end{itemize}
Set $\sigma$ is usually called a $\sigma$-algebra on $X$, or a Borel space on $X$. A measurable
function from a measurable space $\tuple{X, \sigma_X}$ to a measurable space $\tuple{Y, \sigma_Y}$ is a
function $\mor{f}{X}{Y}$ such that if $Y_0 \in \sigma_Y$ then $f^{-1}[Y_0] \in \sigma_X$. A sub-probability
measure $\mu$ on a measurable space $\sigma$, is a countably additive function $\mor{\mu}{\sigma}{[0, 1]}$, i.e.~for every countable family $(X_i)_{i\in \mathcal{N}}$ in $\sigma$ of pairwise disjoint sets $X_i \neq X_j$ whenever $i \neq j$, we have that:
$\mu(\bigcup_{i\in \mathcal{N}} X_i) = \sum_{i\in \mathcal{N}} \mu(X_i)$.
We call a sub-probability measure $\mu$ a probability measure if $\mu(X) = 1$.
Of a special interest are measurable spaces $\tuple{X, \sigma}$ whose $\sigma$-algebra $\sigma$ is the
full powerset on $X$. The reason is that every function from $X$ is measurable according to such $\tuple{X, \sigma}$.

Let us denote by $\catw{m}(X)$ the set of all probability measures on the powerset $\mathcal{P}(X)$ of $X$. Note, that the structure of $\catw{m}(X)$ for an arbitrary $X$ may be difficult to describe (i.e.~this structure highly depends on the foundations of the ambient set theory, in particular, it depends on the existence of large cardinals). Nonetheless, if $X$ is at most countable then one may easily describe the structure of $\catw{m}(X)$, i.e.~every measure $\mu$ on $\mathcal{P}(X)$ is discrete in the sense that $\mu$ is fully determined by its values on singletons $\{x\} \in \mathcal{P}(X)$. Therefore, every such a measure is tantamount to a function $\mor{\mu \downarrow X}{X}{[0, 1]}$ such that $\sum_{x\in X} (\mu \downarrow X)(x) = 1$. The next theorem extends this characterisation to all definable sets in $\classifying{A}$ for an $\omega$-categorical $\omega$-stable structure $\struct{A}$.

\begin{theorem}[Characterisation of measures on a definable set]\label{t:definable:measure}
    Let $\struct{A}$ be an $\omega$-categorical and $\omega$-stable structure. For every definable set $X$ in $\classifying{A}$ every $A_0$-supported probability  measure on $X$ is a finite combination of $A_0$-supported ultrafilters on $X$, i.e.~every $A_0$-equivariant measure $\mor{\mu}{\mathcal{P}(X)}{[0, 1]}$ is of the form $\mu = r_1 p_1 + r_2 p_2 + \cdots + r_n p_n$ for some real numbers $0 < r_i \leq 1$, and ultrafilters $p_i \in \overline{X}$ for $1 \leq i \leq n$.
\end{theorem}
We claim that $(\lambda_p)_{p \in \overline{X}}$ generate measures on $X$ in the following sense: every assignment $\mor{f}{\overline{X}}{[0, 1]}$ extends to at most one measure $\mu$ with $\mu(\lambda_p) = f(p)$. Moreover, if the measure is $A_0$-definable then $f$ must be also $A_0$-definable. This will give us an upper-bound on the size of structure $\catw{m}(X)$.
We prove the claim by induction on $k$ such that $\lambda_p$ for $p \in \overline{X}_{(k)}$ generate measures on $Y$ that are not contained in ultrafilters outside of $\overline{X}_{(k)}$. Let us take any subset $Y \subseteq X$. If $Y$ does not belong to any non-principal ultrafilter, then $Y$ is finite and $\mu(Y) = \sum_{y \in Y} \{y\}$, where the singletons $\{y\}$ are $\lambda_y$ for $y$ treated as principal ultrafilter. Assume that the theorem is true for all $k' < k$ and the maximal $p$ such that $Y \in p$ belongs to $\overline{X}_{(k)}$. Then by Theorem~\ref{t:free:space} for the $2$-element field $Y = \lambda_{p_1} \oplus \lambda_{p_2} \oplus \cdots \oplus \lambda_{p_n} \oplus Y_0$ for $p_i \in \overline{X}_{(k)}$ and $Y_0$ such that if $Y_0 \in p$ then $p \in \overline{X}_{(ki)}$ for $k' < k$. Observe, that by our choice of $\lambda_p$ we have that $\lambda_{p_i} \cap \lambda_{p_j}$ does not belong to a type $p \in \overline{X}_{(k')}$ for $k' \geq k$. Therefore, we can write the measure on each pair as $\mu(\lambda_{p_{i}} \oplus \lambda_{p_{i+1}}) = \mu(\lambda_{p_{i}}) + \mu(\lambda_{p_{i+1}}) - 2\mu(\lambda_{p_i} \cap \lambda_{p_j})$. Thus,we can restrict to the case with at most one $p_1$, i.e.~$Y = \lambda_{p_1} \oplus Y_0$. But then: $\mu(Y) = \mu(\lambda_{p_1}) - 2\mu(\lambda_{p_1} \cap Y) + \mu(Y)$, where $\mu(\lambda_{p_1} \cap Y)$ and $\mu(Y)$ are given by the inductive hypothesis.

It is possible to impose some restrictions on $\mor{f}{\overline{X}}{[0, 1]}$ to induce at least one measure on $\mathcal{P}(X)$ and give a direct proof of Theorem~\ref{t:definable:measure} along this line. Instead, we use a characterisation of measures in $\omega$-stable structures from \cite{chernikov2022definable} (Remark~2.2), which is originally due to H.J~Keisler.
\begin{theorem}[Keisler on Keisler measures in $\omega$-stable structures]\label{t:stable:keisler:measure}
    Let $T$ be $\omega$-stable and $\mu$ a probability measure over the Monster model $\mathcal{U}$ of $T$. Then $\mu = \sum_{i = 0}^\infty r_i p_i$  for $p_i \in S_n(U)$ and $r_i \in [0, 1]$ such that $\sum_{i = 0}^\infty r_i = 1$ .
\end{theorem}
Because we are working with $\omega$-categorical structure $\struct{A}$, we can replace the Monster model $\mathcal{U}$ with $\struct{A}$. So the theorem says that every probability measure on $A^n$ is an infinite positive convex combination of countably many ultrafilters on $A^n$. But in case of $\omega$-categorical structure this result can be improved to \emph{finitely} many ultrafilters.

\begin{proof}[Proof of Theorem~\ref{t:definable:measure}]
    Let an $A_0$-supported measure $\mor{\mu}{\mathcal{P}(A^n)}{[0, 1]}$ be given. By Theorem~\ref{t:stable:keisler:measure} we have that: $\mu = \sum_{i = 0}^\infty r_i p_i$. 
    Because $\overline{A^n}$ has finitely many orbits (by Lemma~\ref{l:nice:exists}), if the number of non-zero $r_i$ is infinite, then there must be an orbit (using the notation from the proof of Theorem~\ref{t:free:space}) in $\overline{A^n}_{(k)}$ that contains an infinite sequence of $p_{i_j}$ such that $r_{i_j}$ is strictly decreasing to zero. Let us assume that $k$ is the smallest number with this property. This means, that there are only finitely many ultrafilters $p_{k_1}, p_{k_2}, \cdots, p_{k_m}$ belonging to $\overline{A^n}_{(k')}$ for $k' > k$.  Consider the measures of $\lambda_{p_{i_j}}$. If $p \in \overline{A^n}_{(k)}$ then $\lambda_{p_{i_j}} \in p$ only if $p = p_{i_j}$, because $\lambda_{p_{i_j}}$ isolates $p_{i_j}$ in $\overline{A^n}^{(k)}$. Therefore:
    $\mu(\lambda_{p_{i_j}}) = r_{i_j} + r_{k_1}p_{k_1}(\lambda_{p_{i_j}}) + r_{k_2}p_{k_2}(\lambda_{p_{i_j}}) + \cdots + r_{k_m}p_{k_m}(\lambda_{p_{i_j}})$.
    But there are only $2^m$ distinct subsets of $r_{k_1}, r_{k_2}, \dotsc, r_{k_m}$, thus they can produce at most $2^m$ distinct values. Therefore, $\mu(\lambda_{p_{i_j}})$ must take infinitely many distinct values, what contradicts the fact $\mu$ is finitely supported. For general equivariant set $X$, observe that by elimination of imaginaries $X \subseteq A^n$ for some $n$ and $A^n \setminus X$ is also equivariant. Therefore, by additivity, every measure on $X$ is just a restriction of a measure on $A^n$. 
\end{proof}

\subsection{Measurable spaces, random variables and Giry monad}

Measures are interesting because we can integrate functions with respect to them. In case of probability measures, we also use the term ``expected value''.
A measurable space $\tuple{A, \sigma}$ with a distinguished probability measure $\mor{\mu}{\sigma}{[0, 1]}$ is called a \emph{probability space} and denoted by $\tuple{A, \sigma, \mu}$. A measurable function from a probability space $\tuple{A, \sigma, \mu}$ to the canonical measurable space of real numbers $\mathcal{R}$ is called a \emph{random variable}. Notice, that if $\sigma = \mathcal{P}(A)$ then every function $\mor{X}{A}{\mathcal{R}}$ is measurable, therefore every such $X$ may be treated as a random variable after fixing a probability measure $\mu$ on $A$.    
Let $X$ be a random variable on a definable probability space $\tuple{A, \sigma, \mu}$. Then the expected value of $X$ will be denoted as $E[X]$ or  $\int_{a\in A} X(a) d\mu$ and defined as $E[X] = \sum_{r \in \mathit{Im}(X)} r \mu(\{a \in A \colon X(a) = r\})$.

Observe that in case $A$ is definable, the image of $X$ is finite, so the above definition is sound. The notion of expected value gives a convenient way to define more advanced concepts. Till the end of the section we will restrict to probability measures on the full $\sigma$-algebras. Consider probability measures $\mor{p}{\mathcal{P}(Q)}{[0, 1]}$ and $\mor{q}{\mathcal{P}(W)}{[0, 1]}$ on definable spaces $Q, W$. We can define a function $\mor{h}{\mathcal{P}(Q\times W)}{[0, 1]^W}$ as the transposition of the following composition:
$\mathcal{P}(Q\times W) \times W \approx \mathcal{P}(Q)^W \times W \overset{\epsilon}\rightarrow \mathcal{P}(Q) \overset{p}\rightarrow [0, 1]$.
Then: $\mu(P_0) = \int_{x \in W} h(P_0)(x) dq$ defines a probability measure on $Q \times W$. There is a symmetric way to define a measure on $Q\times W$ -- i.e.~by swapping the order in the product. That is, define $\mor{k}{\mathcal{P}(Q\times W)}{[0, 1]^Q}$ as the transposition of $\mathcal{P}(Q\times W) \times Q \approx \mathcal{P}(W)^Q \times Q \overset{\epsilon}\rightarrow \mathcal{P}(W) \overset{q}\rightarrow [0, 1]$
and then: $\mu'(P_0) = \int_{x \in Q} k(P_0)(x) dp$. It follows that for $\omega$-stable theories measures $\mu$ and $\mu'$ coincide. Therefore, we can speak of \emph{the product measure}.
\begin{remark}[Product measures in non-stable theories]
    For non-stable theories $\mu$ and $\mu'$ can be different. For example, consider ultrafilters $q=p=0^{+}$ on $Q$ in DLO (see Remark~\ref{r:stability:is:necessary}) and treat them as probability measures. By the definition we have: $h(R_0)(x) = p(\{a \in Q \colon R_0(a,x) \}) \Leftrightarrow (0, c) \times \{x\} \subseteq R_0$ for some positive $c$, and so:
$\mu(R_0)=\int_{x \in Q}\;\exists_{c>0}\; [(0, c) \times \{x\} \subseteq R_0] dq = 0^{+}(\{x \in Q \colon \exists_{c>0} [(0, c) \times \{x\} \subseteq R_0]\}) = \exists_{d>0}\; (0, d) \subseteq \{x \in Q \colon \exists_{c>0} (0, c) \times \{x\} \subseteq R_0\}$. Therefore, $R_0 \in \mu$ if and only if $R_0$ contains a triangle with vertices $(0, 0), (d, d), (0, d)$ for some $d>0$. Similarly, $k(R_0)(a) = p(\{x \in Q \colon R_0(a,x) \} \Leftrightarrow \{a\} \times (0, c) \subseteq R_0$ for some $c>0$, and: $\mu'(R_0) = \int_{a \in Q} k(R_0)(a) dp = p(\{a \in Q \colon k(R_0)(a)\}) = 0^{+}(\{a \in Q \colon \exists_c \{a\} \times (0, c) \subseteq R_0\}) = \exists_{d>0}\; (0, d) \subseteq \{a \in Q \colon \exists_{c>0} \{a\} \times (0, c) \subseteq R_0\}$. Therefore, $R_0 \in \mu'$ if and only if $R_0$ contains a triangle with vertices $(0, 0), (d, d), (d, 0)$ for some $d>0$. Intuitively, when constructing the product measure, we have to favour one of the directions, because every 2-dimensional set decomposes on subsets defined by $x > y$, $x < y$ and $x=y$. 
\end{remark}

The construction $\catw{m}(X)$ extends to a functor on the category of measurable spaces and measurable functions. Moreover, this functor can be equipped with the usual structure of a Giry monad \cite{giry2006categorical}. The monad is strong and so gives the structure of an internal monoid on maps $X \rightarrow \catw{m}(x)$. The Kleisli unit $\mor{\eta_X}{X}{\catw{m}(X)}$ assigns to an element $x \in X$ the principal ultrafilter generated by $x$ and the Kleisli composition of $\mor{f}{X}{\catw{m}(Y)}$ with $\mor{g}{Y}{\catw{m}(Z)}$ is defined as the composition of Markov kernels (see:  \cite{giry2006categorical}), i.e.~$g(f(x))(Z_0) = \int_{y\in Y} g(y)(Z_0)df(x)$. In case $X, Y, Z$ are definable this reduces to $g(f(x))(Z_0) = \sum_{r \in[0, 1]} r \cdot f(x)(\{y \in Y \colon g(y)(Z_0) = r\})$, where the extension function $\mor{\overline{g}}{\catw{m}(Y)}{\catw{m}(Z)}$ is just: $\overline{g}(\mu)(Z_0) = \sum_{r \in[0, 1]} r \cdot \mu(\{y \in Y \colon g(y)(Z_0) = r\})$ and the summations are effectively finite by Theorem~\ref{t:definable:measure}.

\section{Conclusions and future work}\label{sec:conclusions}
The paper investigates properties of Stone-Čech compactification of discrete spaces in various models of Zermelo-Fraenkel Set Theory with Atoms. This theme is interesting from both theoretical and practical perspectives.

From the theoretical point of view, we show that in ZFA over certain $\omega$-categorical structures, the Stone-Čech compactification of a definable set is definable -- this is the case $\omega$-stable structures (Theorem~\ref{t:stone:cech:compactification}) and rational numbers with their natural ordering $\struct{Q}$ (Theorem~\ref{t:definable:types:dlo}), but not the case of random graphs (Example~\ref{e:types:in:graphs}), nor the polar geometry (similar argument). We conjecture that this is true exactly for the class of NIP structures. Moreover, for $\omega$-stable structures, the process of Stone-Čech compactification gives an explicit description of a basis of the dual vector space $V^*$ to a vector space $V$ with a definable basis $\Lambda$ -- i.e.~the basis of $V^*$ is Stone-Čech compactification $\overline{\Lambda}$ of $\Lambda$ (Theorem~\ref{t:dual:basis}). This does not hold for non-$\omega$-stable structures (Remark~\ref{r:stability:is:necessary}), but still can help -- e.g.~for $\struct{Q}$ such a basis exists and is a definable subset of $\overline{\Lambda}$ (Theorem~\ref{t:definable:types:dlo} together with Theorem~\ref{t:rational:basis}). We conjecture that this is generally true for the class of NIP structures. The existence of a definable basis for such a dual space is the main ingredient in proving that the category of vector spaces on definable sets is monoidal closed (Theorem~\ref{t:monoidal:closed}). This is quite remarkable, because the category of definable sets is usually \emph{not} closed. We believe that studying further properties of the category of vector spaces over definable basis will lead to many interesting and practical results. E.g.~the authors of \cite{BKM21} showed that vector spaces over definable basis in $\struct{Q}$ are of finite length. Can this be generalised to other ZFA, etc.? In Section~\ref{sec:probability} we show that probability measures on definable sets are quite well behaved and allow for developing a bit of probability theory. Interestingly, probability measures on a definable set $X$ in an $\omega$-stable structure are \emph{finite} convex combinations of mass-measures on the Stone-Čech compactification $\overline{X}$ of $X$ (Theorem~\ref{t:definable:measure}). This is analogous to the classical fact that probability measures on a finite set are finite convex combinations of mass-measures on it (or on its Stone-Čech compactification) and gives an explicit description of the structure of such measures. We believe that further properties of internal measures should be studied.

From the practical point of view, we answer some open question raised in \cite{BKM21} and in \cite{10.1145/3531130.3533333} and give smoother and more general results for existence of dual basis. The existence of these basis is the main tool in \cite{10.1145/3531130.3533333} to prove solvability of systems of definable linear equations. The bare fact that the Stone-Čech compactification of a definable set is definable allows us to show that register machines extended with the ability to erase the content of their registers can be reduced to the classical register machines (Theorem~\ref{t:expressive:ultra:automata}). The fact that the category of vector spaces on definable basis is monoidal closed allows us to slightly generalise the construction of weighted automaton without changing the concept of the recognised language and obtain a general equivalence between languages recognised by weighted automata and languages recognised by linear monoids (Theorem~\ref{t:monoid}). Nonetheless, perhaps the most interesting application of the concepts developed in this paper is the definition and characterisation of general probabilistic register machines -- according to our knowledge, this concept has not been studied in such a generality before. It turns out, that just like in the classical setting, i.e.~finite case, they can be embedded in weighted automata over the Stone-Čech compactification of the states (Theorem~\ref{t:probabilistic:automaton}). Finally, in \cite{sabok2021probabilistic} probabilistic semantics for lambda calculus with fresh names are studied. We believe that the natural setting for these semantics is ZFA with internal measures.

\section*{Acknowledgement}
This research was supported by the National Science Centre, Poland, under projects 2018/28/C/ST6/00417.

\printbibliography

@inproceedings{giry2006categorical,
  title={A categorical approach to probability theory},
  author={Giry, Michele},
  booktitle={Categorical Aspects of Topology and Analysis: Proceedings of an International Conference Held at Carleton University, Ottawa, August 11--15, 1981},
  pages={68--85},
  year={2006},
  organization={Springer}
}

@article{pincus1977definability,
  title={Definability of measures and ultrafilters},
  author={Pincus, David and Solovay, Robert M},
  journal={The Journal of Symbolic Logic},
  volume={42},
  number={2},
  pages={179--190},
  year={1977},
  publisher={Cambridge University Press}
}

@article{blass1997model,
  title={A model without ultrafilters},
  author={Blass, Andreas},
  journal={Bull. Acad. Sci. Polon. Ser. Sci. Math. Astr. Phys.},
  volume={25},
  pages={329--331},
  year={1997}
}

@article{kaminski1994finite,
  title={Finite-memory automata},
  author={Kaminski, Michael and Francez, Nissim},
  journal={Theoretical Computer Science},
  volume={134},
  number={2},
  pages={329--363},
  year={1994},
  publisher={Elsevier}
}

@article{chernikov2022definable,
  title={Definable convolution and idempotent Keisler measures},
  author={Chernikov, Artem and Gannon, Kyle},
  journal={Israel Journal of Mathematics},
  volume={248},
  number={1},
  pages={271--314},
  year={2022},
  publisher={Springer}
}

@INPROCEEDINGS{BKM21, 
author={M. {Boja\'nczyk} and B. {Klin} and M. {Moerman}}, 
booktitle={2021 36th Annual ACM/IEEE Symposium on Logic in Computer Science (LICS)}, 
title={Orbit-finite-dimensional vector spaces and weighted register automata}, 
year={2021}, 
volume={}, 
number={}, 
pages={1-13}, 
doi={10.1109/LICS52264.2021.9470634}
}

@book{marker2006model,
  title={Model theory: an introduction},
  author={Marker, David},
  volume={217},
  year={2006},
  publisher={Springer Science \& Business Media}
}

@book{hodges1993model,
  title={Model theory},
  author={Hodges, Wilfrid and Wilfrid, Hodges and others},
  year={1993},
  publisher={Cambridge university press}
}

@book{chang1990model,
  title={Model theory},
  author={Chang, Chen Chung and Keisler, H Jerome},
  year={1990},
  publisher={Elsevier}
}

@article{sabok2021probabilistic,
  title={Probabilistic programming semantics for name generation},
  author={Sabok, Marcin and Staton, Sam and Stein, Dario and Wolman, Michael},
  journal={Proceedings of the ACM on Programming Languages},
  volume={5},
  number={POPL},
  pages={1--29},
  year={2021},
  publisher={ACM New York, NY, USA}
}

@misc{casanovas2007stable,
  title={Stable and Simple theories (Lecture Notes)},
  author={Casanovas, Enrique},
  year={2007},
  publisher={Universidad de Barcelona}
}

@misc{harry2021,
    title        = {If a vector space has a basis then its dual vector space has a basis},
    author       = {Harry West},
    year         = 2021,
    note         = {\url{https://mathoverflow.net/questions/395996/if-a-vector-space-has-a-basis-then-its-dual-vector-space-has-a-basis}}
}

@inproceedings{10.1145/3531130.3533333,
author = {Ghosh, Arka and Hofman, Piotr and Lasota, Slawomir},
title = {Solvability of orbit-finite systems of linear equations},
year = {2022},
isbn = {9781450393515},
publisher = {Association for Computing Machinery},
address = {New York, NY, USA},
url = {https://doi.org/10.1145/3531130.3533333},
doi = {10.1145/3531130.3533333},
booktitle = {Proceedings of the 37th Annual ACM/IEEE Symposium on Logic in Computer Science},
articleno = {11},
numpages = {13},
location = {Haifa, Israel},
series = {LICS '22}
}

@book{simon2015guide,
  title={A guide to NIP theories},
  author={Simon, Pierre},
  year={2015},
  publisher={Cambridge University Press}
}

@book{kelly1982basic,
  title={Basic concepts of enriched category theory},
  author={Kelly, Max},
  volume={64},
  year={1982},
  publisher={CUP Archive}
}

@book{tent2012course,
  title={A course in model theory},
  author={Tent, Katrin and Ziegler, Martin},
  number={40},
  year={2012},
  publisher={Cambridge University Press}
}

@article{cherlin1985no,
  title={$\aleph$0-Categorical, $\aleph$0-Stable Structures},
  author={Cherlin, Gregory and Harrington, Leo and Lachlan, Alistair H},
  journal={Annals of Pure and Applied Logic},
  volume={28},
  number={2},
  pages={103--135},
  year={1985},
  publisher={North-Holland}
}

@inproceedings{licsMRP,
  author    = {Michal R.~Przybylek},
  title     = {Beyond sets with atoms: definability in first order logic},
  booktitle = {https://arxiv.org/abs/2003.04803},
  year      = {2020},
}

@article{blass2011partitions,
  title={Partitions and permutation groups},
  author={Blass, Andreas},
  journal={ Model Theoretic Methods in Finite Combinatorics },
  volume={558},
  pages={453--466},
  year={2011}
}

@incollection{jech1977axiom,
  title={About the axiom of choice},
  author={Jech, Thomas J},
  booktitle={Studies in Logic and the Foundations of Mathematics},
  volume={90},
  pages={345--370},
  year={1977},
  publisher={Elsevier}
}

@book{jech2008axiom,
  title={The Axiom of Choice},
  author={Jech, T.J.},
  isbn={9780486466248},
  lccn={2008014612},
  series={Dover Books on Mathematics Series},
  year={2008},
  publisher={Dover Publications}
}

@article{mostowski1939unabhangigkeit,
  title={{\"U}ber die Unabhangigkeit des Wohlordnungssatzes vom Ordnungsprinzip},
  author={Mostowski, Andrzej},
  year={1939}
}

@book{halbeisen2017combinatorial,
  title={Combinatorial Set Theory: With a Gentle Introduction to Forcing},
  author={Halbeisen, Lorenz J},
  year={2017},
  publisher={Springer}
}

@book{maclane2012sheaves,
  title={Sheaves in geometry and logic: A first introduction to topos theory},
  author={MacLane, Saunders and Moerdijk, Ieke},
  year={2012},
  publisher={Springer Science \& Business Media}
}

@article{johnstone2003sketches,
  title={Sketches of an elephant: A topos theory compendium-2 volume set},
  author={Johnstone, Peter T},
  journal={Oxford University Press, ISBN-10:. ISBN-13: 9780198524960},
  pages={1288},
  year={2003}
}

\appendix

\section{One theorem from the introduction}\label{sec:app:intro:theorem}

\begin{theorem}[Ultrafilters in ZFA over $\omega$-categorical $\omega$-stable structures]\label{t:ultrafilters:zfa}
Let $A$ be a non-trivial $\omega$-categorical and $\omega$-stable structure. Then:
\begin{enumerate}
    \item Boolean Prime Ideal Theorem does not hold in $\classifying{A}$
    \item for every infinite set $X$ in $\classifying{A}$ there is a non-principal ultrafilter on $\mathcal{P}(X)$.
\end{enumerate}
\end{theorem}
\begin{proof}
For (1) recall that BPIT is equivalent over ZF(A) to the compactness theorem of propositional calculus (see, for example \cite{blass2011partitions} or \cite{jech2008axiom}). Let $A$ be the set of atoms, and consider the following set of propositional variables $\mathit{Var} = A^2$ with the following set of propositions:
\begin{itemize}
    \item $\{\neg (\tuple{a, b} \land \tuple{b, a}) \colon \tuple{a, b}, \tuple{b, a} \in \mathit{Var}\}$
    \item $\{\tuple{a, b} \lor \tuple{b, a} \colon \tuple{a, b}, \tuple{b, a} \in \mathit{Var} \land a \neq b\}$
    \item $\{\tuple{a, b} \land \tuple{b, c} \rightarrow \tuple{a, c} \colon \tuple{a, b}, \tuple{b, c}, \tuple{a, c} \in \mathit{Var}\}$
\end{itemize}
Intuitively, the sets of propositions say that there exists a strict linear ordering on $A$. Because, every finite subset of the sets of propositions is satisfiable (i.e.~there are definable orders in $A$ of any finite length), by the compactness theorem for propositional calculus, the whole set is satisfiable, which means that there is a strict linear order on $A$. But this contradicts $\omega$-stability of structure $\struct{A}$. 

For (2), without loss of generality we assume that $X$ is equivariant. Therefore, it is a disjoint union of $I$ equivariant sets $(X_i)_{i \in I}$ consisting of its orbits. If $I$ is finite, then by $\omega$-categoricity of $\struct{A}$ set $X$ is definable and we have the structure theorems (Theorem~\ref{t:stone:cech:compactification}, see also Remark~\ref{r:non:principal:ultrafilters:exist}) for the set of ultrafilters on $X$. So, let us assume that $I$ is infinite. Then, assuming AC (or at least BPIT) in the external (meta-)mathematics, there is a non-principal ultrafilter $\mu_I$ on $I$. Let us associate with every $i \in I$ an equivariant ultrafilter (principal or not) $\mu_i$ on $X_i$. Then, we may define an ultrafilter $\mu$ on $X$ as the Fubini-like product:
    $$A \in \mu \Leftrightarrow \{i \in I \colon A \downarrow X_i \in \mu_i\} \in \mu_I $$
    for every $A \subseteq X$, where  $A \downarrow X_i$ is the restriction of $A$ to $X_i$. To see that $\mu$ is symmetric, let us consider any permutation $\pi$. We have:
    $$\pi(A) \in \mu \Leftrightarrow \pi(\{i \in I \colon A \downarrow X_i \in \mu_i\}) \in \mu_I \Leftrightarrow \{i \in I \colon \pi(A \downarrow X_i) \in \mu_i\} \in \mu_I$$
    Obviously, $A \downarrow X_i \subseteq X_i$ and $X_i$ consists of a single orbit, thus $\pi(A \downarrow X_i) \subseteq \pi(X_i) = X_i$. Therefore, by equivariance of $\mu_i$ we have that: $\pi(A \downarrow X_i) \in \mu_i \Leftrightarrow A \downarrow X_i \in \mu_i$ and so: $A \in \mu \Leftrightarrow \pi(A) \in \mu$. It is a rutine to check that such defined $\mu$ is an ultrafilter:
    \begin{itemize}
        \item $\emptyset \in \mu \Leftrightarrow \{i \in I \colon \emptyset \in \mu_i\} \in \mu_I \Leftrightarrow \emptyset I$, so $\emptyset$ does not belong to $\mu$ because it does not belong to $I$
        \item if $A, B \in \mu$ then both  $I_A = \{i \in I \colon A \downarrow X_i \in \mu_i\}$ and $I_B = \{i \in I \colon A \downarrow X_i \in \mu_i\}$ and because $I$ is an ultrafilter $I_A \cap I_B \in \mu_I$; but $I_A \cap I_B = \{i \in I \colon A \downarrow X_i \in \mu_i \land B \downarrow X_i \in \mu_i\}$ and because $\mu_i$ is an ultrafilter $(A \downarrow X_i) \cap B (\downarrow X_i) = (A\cap B)\downarrow X_i \in \mu_i$, so $A\cap B \in \mu$
        \item by definition $(X \setminus A) \in \mu \Leftrightarrow  \{i \in I \colon (X \setminus A) \cap X_i \in \mu_i\} \in \mu_I$; we have however, $(X \setminus A) \cap X_i = X_i \setminus A = X_i \setminus (A \cap X_i)$, so $(X \setminus A) \cap X_i \in \mu_i \Leftrightarrow A \cap X_i \not\in \mu_i$; but $\{i \in I \colon A \cap X_i \not\in \mu_i\} \in \mu_I \Leftrightarrow \{i \in I \colon A \cap X_i \in \mu_i\} \not\in \mu_I \Leftrightarrow A \not\in \mu$
    \end{itemize}
\end{proof}

\section{On polar geometry}\label{sec:app:polar}
Let $V$ and $V^*$ be two disjoint (necessarily isomorphic) free $\aleph_0$-dimensional vector space over a finite field $\mathcal{F}$.  Let $\Phi$ be a bilinear map $\mor{\Phi}{V \times V^*}{L}$ satisfying the following axiom: (Space Extension Axiom) for every finite sequence of linearly independent vectors $\{v_i\}_{1 \leq i \leq k}$ from $V$ (resp.~$V^*$) together with a sequence of scalars $\{r_i \in \mathcal{F}\}_{1 \leq i \leq k}$ and a finite set of vectors $W \subset V^*$ (resp.~$W \subset V$)  there exists $w \notin W$ such that $\Phi(v_i, w) = r_i$ (resp.~$\Phi(w, v_i) = r_i$).
The polar geometry over finite field $\mathcal{F}$ is the structure $\struct{G}_\mathcal{F} = \tuple{V \cup V^*, B(-), \Phi, +, (-)r}$, where $+$ and $(-)r$ are interpreted separately on each of vector spaces $V,  V^*$ and predicate $B$ distinguishes vectors $V$ from v$V^*$, i.e.~$B(x) \Leftrightarrow x \in V$. The polar geometry is $\omega$-categorical (by the usual back-and-forth argument), but not $\omega$-stable. It has however, a good notion of independence (i.e.~it is a simple theory \cite{casanovas2007stable}), which will be important for the proof of the next theorem, which is essentially due to Harry West \cite{harry2021}. In the below, we shall also use the classic fact that polar geometry has the intersection property of algebraically closed supports (i.e.~a set is algebraically closed if it contains elements of every finite set definable in it). For a prime number $p$, let us denote by $\mathcal{F}_p$ the finite field of characteristic $p$.

\begin{theorem}[Dual basis in polar geometries]\label{t:polar:counterexample}
    For every prime number $p$, there is an equivariant definable set $X$ in $\classifying{\struct{G}_{\mathcal{F}_p}}$ such that $\mathcal{F}_p^X$ does not have a basis.
\end{theorem}

The proof is by contradiction. Let us suppose that $\Lambda$ is a $V_0 \cup V_0^*$-supported basis of $\mathcal{F}_p^V$.  We shall assume that $V_0 \cup V_0^*$ is algebraically closed in $\struct{G}_{\mathcal{F}_p}$. Let $\alpha, \beta \in V^*$ be linearly independent over $V_0^*$, i.e.~$\alpha+V_0^*$ is linearly independent from $\beta + V_0^*$. Then they must be linearly independent over $\emptyset$ and so $\mathit{dim}(\mathit{span}(\alpha, \beta)) = 2$.
Let us denote by $L = \{\mathit{span}(\alpha + k\beta) \colon k \in \mathcal{F}_p\} \cup \{\mathit{span}(\beta)\}$ the set of all one-dimensional subspaces of $\mathit{span}(\alpha, \beta)$. For every one dimensional subspace $l \in L$ define the following function $\mor{f_l}{V}{\mathcal{F}_p}$:
    $$f_l(v) = \begin{cases}
0 \;\; \textit{if} \; \forall_{\gamma \in l}\; \Phi(v, \gamma) = 0,\\
1 \;\; \textit{otherwise}
\end{cases}$$
Notice that by definition $f_l$ is $l$-supported. The functions $f_l$ are chosen in such a way that they sum up to the zero function.
\begin{lemma}
    For  $\mor{f_l}{V}{\mathcal{F}_p}$ defined as in the above, we have that: $\sum_{l \in L} f_l \equiv 0$.
\end{lemma}
\begin{proof}
For any $v \in V$ let us consider the functional $\mor{h_v}{\mathit{span}(\alpha, \beta)}{\mathcal{F}_p}$ defined as the restriction of $\Phi(v, -)$ to $\mathit{span}(\alpha, \beta))$, i.e.~$h_v = \Phi(v, -)\downarrow \mathit{span}(\alpha, \beta)$. By the classical rank-nullity theorem we have that:
$$\mathit{dim}(\mathit{Ker}(h_v)) + \mathit{dim}(\mathit{Im}(h_v) = \mathit{dim}(\mathit{span}(\alpha, \beta)) = 2$$
and because $\mathit{Im}(\Phi(v, -)) \subseteq \mathcal{F}_p$ we have that $\mathit{dim}(\mathit{Ker}(h_v))$ is either $1$ or $2$. In case $\mathit{dim}(\mathit{Ker}(h_v)) = 2$ it must be that $\mathit{Ker}(h_v) = \mathit{span}(\alpha, \beta)$ and then $f_l(v) = 0$ for every $l \in L$ so $\sum_{l \in L} f_l(v) = 0$. In case $\mathit{dim}(\mathit{Ker}(h_v)) = 1$ it must be that $\mathit{Ker}(h_v) = l$ for some $l \in L$, and since $|L \setminus \{l\}| = p$
we have:
$$\sum_{l' \in L} f_{l'}(v) = f_l(v) + \sum_{l' \in L \setminus \{l\}} f_{l'}(v) = 0 + p\cdot1 = 0$$    
\end{proof}
Let $F_S \subseteq \mathcal{F}_p^V$ be the set of all functions $V \rightarrow \mathcal{F}_p$ supported by $V_0 \cup \mathit{span}(V_0^* \cup S)$.
\begin{lemma}
    For $l \neq l' \in L$ we have that if $f \in F_l \cap F_{l'}$ then $f \in F_\emptyset$.
\end{lemma}
\begin{proof}
    $f \in F_l \cap F_{l'}$ if and only if $f$ is supported by both $V_0 \cup \mathit{span}(V_0^* \cup l)$ and $V_0 \cup \mathit{span}(V_0^* \cup l')$. Because $\alpha$ and $\beta$ are independent over $V_0^*$ every pair $\tuple{v, w}$ of non zero vectors $v \in l$ and $w \in l'$ for distinct $l, l' \in L$ is independent over $V_0^*$. Therefore, $\mathit{span}(V_0^* \cup l) \cap \mathit{span}(V_0^* \cup l') = V_0^*$ and by the intersection property of the acl-supports $f \in F_\emptyset$.
\end{proof}

\begin{lemma}
    For every $l \in L$ we have that $f_l \not\in F_\emptyset$.
\end{lemma}
\begin{proof}
Let us consider a non-zero vector $w \in l$. By the extension axioms, for any $k \in \mathcal{F}_p$ there is $v_k \not\in V_0$ such that for every $\gamma \in V_0^*$ we have that $\Phi(v_k,\gamma) = 0$ and $\Phi(v_k, w) = k$. By the definition all $v_k$ are in the same $V_0 \cup V_0^*$-orbit, but for $k \neq k'$ we have that: $f_l(v_k) = k \neq k' = f_l(v_{k'})$, therefore $f_l$ is not $V_0 \cup V_0^*$-supported, i.e.~$f_l \not\in F_\emptyset$.    
\end{proof}

Now, let us make the following two simple remarks.
\begin{remark}[Coefficients are definable]\label{r:definable:coefficients}
    Let $\Lambda$ be an $A_0$-supported basis of a vector space $V$. For every $\lambda \in \Lambda$ denote by $\mor{c_\lambda}{V}{\mathcal{K}}$ the function that sends a vector $v \in V$ to its $\lambda$-coefficient. Then the set consisting of all coefficient functions $C = \{c_\lambda \colon \lambda \in \Lambda \}$ is supported by $A_0$. So see this, let us first define the set $\tilde{C} = \{\tilde{c}_\lambda \in \mathcal{K}^\Lambda \colon \tilde{c}_\lambda(\lambda') = [\lambda = \lambda'] \}$. The definition is in the terms of basis $\Lambda$, therefore the set is $A_0$-supported. Because the free vector space monad is equivariant it preserves the supports and so $C$ is also $A_0$-supported.
\end{remark}
\begin{remark}\label{r:definable:elements}
    Suppose that an $\omega$-categorical structure $\struct{A}$ has the intersection property for algebraically closed supports. If a finite set $X$ is supported by a minimal algebraically closed set $S$ then every $x \in X$ is supported by $S$. To see this, let us write $Z = \bigcup_{x\in X} \mathit{supp}(x)$. We want to show that $Z \subseteq S$. For contradiction let us assume that there exists $z \in \mathit{acl}(Z)$ such that $z \not\in S$. Then, because $S$ is algebraically closed the $S$-orbit $[z]_S$ of $z$ must be infinite (otherwise $z$ would be algebraic over $S$). Because $Z$ is finite and $\struct{A}$ is $\omega$-categorical $\mathit{acl}(Z)$ must be finite and so there is $\pi_S$ such that $\pi_S(z) \not\in \mathit{acl}(Z)$. On the other hand $\pi_S^{-1}(X) = X$ so every element $x \in X$ is also supported by $\pi_S^{-1}(Z)$ and by the intersection property for supports, it is supported by $\mathit{acl}(Z) \cap \pi_S^{-1}(\mathit{acl}(Z))$. Because $\pi_S(z) \not\in \mathit{acl}(Z)$ we have that $z \not\in \pi_S^{-1}(\mathit{acl}(Z)) = \mathit{acl}(\pi_S^{-1}(Z))$ and so $z \not\in \pi_S^{-1}(\mathit{acl}(Z)) \cap \mathit{acl}(Z)$. Therefore, $z$ is not in the minimal support of any $x \in X$.   
\end{remark}
Consider the expansion of $f_l$ in basis $\Lambda$, i.e.~$f_l = \sum_{i} c_{\lambda_i}(l) \lambda_i$. The set $I = \{\lambda_i \colon c_{\lambda_i}(l) \neq 0\}$ is $V_0 \cup \mathit{span}(V_0^* \cup l)$-supported by Remark~\ref{r:definable:coefficients} and by Remark~\ref{r:definable:elements} every $\lambda_i \in I$ is $V_0 \cup \mathit{span}(V_0^* \cup l)$-supported, therefore $\lambda_i \in F_l$. Because $f_l \not\in F_\emptyset$ there must be $\lambda_i \in I$ such that $\lambda_i \not\in F_\emptyset$. Therefore, for $l' \neq l$ we have that $\lambda_i \not\in F_{l'}$ and so $c_{\lambda_i}(f_{l'}) = 0$. But then: $0 \neq c_{\lambda_i}(f_{l}) = c_{\lambda_i}(\sum_{l' \in L} f_{l'}) = 0$.

\section{Proofs from the paper}\label{sec:app:proofs}

\begin{proof}[Proof of Lemma~\ref{l:orbit:cardinal:space}]
We claim that for $\lambda \in \mathcal{B}(\kappa)$ and $p \in \overline{X}$ vectors $\mor{(\lambda, p)}{\kappa \times X}{\mathcal{K}}$ defined as $(\lambda, p)(i, x) = \lambda(i)p(x)$ form the basis of $\mathcal{K}^{\kappa \times X}$. Linear independence of the vectors is obvious, therefore let us show that every definable function $\mor{f}{\kappa \times X}{\mathcal{K}}$ is a finite combination of these vectors. Because $f$ is definable it is $A_0$-supported for some finite $A_0$. The crucial observation is that for every $i \in \kappa$ the function $\mor{f(i, -)}{X}{\mathcal{K}}$ must be $A_0$-supported, so there are only finitely many $X_j$ such that $X = \coprod_{1 \leq j \leq n} X_j$ and the restrictions $\mor{f_j(i, -)}{X_j}{\mathcal{K}}$ are constant. Let us denote the constant associated to the pair $i, j$ by $r_{i, j}$. Then $r_{{(-)}, j}$ is a function $\kappa \rightarrow \mathcal{K}$ and as such has a unique decomposition in the basis $\mathcal{B}(\kappa)$, say: $r_{{(-)}, j} = \sum_{1 \leq s_j \leq N} c_{s_j} \lambda_{s_j}$. Moreover, by Theorem~\ref{t:free:space} each $X_j$ has its own decomposition in $\overline{X}$ as $X_j = \sum_{1 \leq t_j \leq M} b_{t_j}p_{t_j}$, where $M, N$ can be chosen to not depend on $j$. So:
$f(i, x) = \sum_{1 \leq j \leq n} f(i, x)X_j = \sum_{1 \leq j \leq n} r_{i, j}X_j = \sum_{1 \leq j \leq n}\sum_{1 \leq s_j \leq N}\sum_{1 \leq t_j \leq n} c_{s_j}b_{t_j} \lambda_{s_j}p_{t_j}$. 
\end{proof}
\begin{proof}[Proof of Theorem~\ref{t:dual:basis}]
Without loss of generality, we shall assume that $X$ is equivariant and $A$ eliminates imaginaries. Let $b$ be the bound on the size of the support of each element $x \in X$. Equivariant set $X$ can be written as a disjoint union of its equivariant orbits $(X_i)_{i \in I}$, where $I$ is a cardinal number. By elimination of imaginaries of $A$, every orbit $X_i$ is isomorphic to an equivariant orbit of $A^b$ and by $\omega$-categoricity of $A$ there are only finitely many of them. Therefore, there are some $X_{i_1}, X_{i_2}, \dotsc, X_{i_n}$ and cardinals $\kappa_1, \kappa_2, \dotsc, \kappa_n$ such that:
$X \approx \coprod_{1 \leq j \leq n} \kappa_j \times X_{i_j}$.
Because, the free vector space functor $F$ preserves colimits, and exponents map colimits to limits:$\mathcal{K}^X \approx \word{Lin}(F(X), \mathcal{K})$  
$\approx \word{Lin}({F(\coprod_{1 \leq j \leq n} \kappa_j \times X_{i_j}}), \mathcal{K}) \approx \prod_{1 \leq j \leq n} \word{Lin}(F(\kappa_j \times X_{i_j}), \mathcal{K}) = \coprod_{1 \leq j \leq n} \word{Lin}(F(\kappa_j \times X_{i_j}), \mathcal{K})$,
where the last equality follows from the fact that finite coproducts coincide with finite products for vector spaces. By Lemma~\ref{l:orbit:cardinal:space} we have that: $\mathcal{K}^{\kappa_j X_{i_j}} \approx F(\mathcal{B}(\kappa_j)\overline{X_{i_j}})$ and so:
$$\mathcal{K}^X \approx \coprod_{1 \leq i \leq n} F(\mathcal{B}(\kappa_j)\times \overline{X_{i_j}}) \approx F(\coprod_{1 \leq i \leq n} \mathcal{B}(\kappa_j)\times \overline{X_{i_j}})$$
Therefore, $\coprod_{1 \leq i \leq n} \mathcal{B}(\kappa_j)\times \overline{X_{i_j}}$ is (isomorphic to) a basis of $\mathcal{K}^X$. Observe that because each $\mathcal{B}(\kappa_j)$ and $\overline{X_{i_j}}$ are equivariant, the constructed basis is equivariant. By Theorem~\ref{t:stone:cech:compactification} each $\overline{X_{i_j}}$ is definable, therefore the support of its elements is bounded by some finite $b_j$. So $\coprod_{1 \leq i \leq n} \mathcal{B}(\kappa_j)\times \overline{X_{i_j}}$ is of a bounded support $\max_{1 \leq i \leq n} b_j$.
Moreover, if $X$ is definable, then the cardinals $\kappa_j$ must be finite, and by Theorem~\ref{t:stone:cech:compactification} the basis consists of finitely many orbits, thus is definable. 
\end{proof}

\end{document}